\documentclass[nolimits,a4paper,envcountsame,final]{llncs}

\usepackage[colorlinks,linkcolor={blue},citecolor={blue},urlcolor={red},breaklinks=true,final]{hyperref}
\usepackage[final]{graphicx}

\usepackage[utf8]{inputenc}

\usepackage{textcomp}

\usepackage{enumitem} 
\setlist{itemsep=0ex}
\setlist[itemize]{label={\small$\bullet$}}

\makeatletter
\newcommand\itemvar[1]{%
 \item[\arabic{enumi}.#1]\def\@currentlabel{\arabic{enumi}.#1}}
\makeatother

\spnewtheorem{thm}[theorem]{Theorem}{\bfseries}{\itshape}
\spnewtheorem{cor}[theorem]{Corollary}{\bfseries}{\itshape}
\spnewtheorem{cnj}[theorem]{Conjecture}{\bfseries}{\itshape}
\spnewtheorem{lem}[theorem]{Lemma}{\bfseries}{\itshape}
\spnewtheorem{lemdefn}[theorem]{Lemma and Definition}{\bfseries}{\itshape}
\spnewtheorem{prop}[theorem]{Proposition}{\bfseries}{\itshape}
\spnewtheorem{defn}[theorem]{Definition}{\bfseries}{\upshape}
\spnewtheorem{rem}[theorem]{Remark}{\bfseries}{\upshape}
\spnewtheorem{notation}[theorem]{Notation}{\bfseries}{\upshape}
\spnewtheorem{expl}[theorem]{Example}{\bfseries}{\upshape}
\spnewtheorem{thmdefn}[theorem]{Theorem and Definition}{\bfseries}{\itshape}
\spnewtheorem{propdefn}[theorem]{Proposition and Definition}{\bfseries}{\itshape}
\spnewtheorem{assumption}[theorem]{Assumption}{\bfseries}{\upshape}
\spnewtheorem{algorithm}[theorem]{Algorithm}{\bfseries}{\upshape}

 \renewenvironment{theorem}{\begin{thm}}{\end{thm}}
 
 \renewenvironment{lemma}{\begin{lem}}{\end{lem}}
 \renewenvironment{proposition}{\begin{prop}}{\end{prop}}
 \renewenvironment{definition}{\begin{defn}}{\end{defn}}
 \renewenvironment{remark}{\begin{rem}}{\end{rem}}
 \renewenvironment{example}{\begin{expl}}{\end{expl}}

\usepackage{todos}
\usepackage{proof}

\input{catprog}

\usepackage{amsmath}
\usepackage{amssymb}
\usepackage{mathrsfs}

\makeatletter
\newcommand\Label[1]{&\refstepcounter{equation}(\theequation)\ltx@label{#1}&}
\makeatother

\newcommand{\klstar}{\sharp}  			
\newcommand{\istar}{\dagger}  			
\newcommand{\iistar}{\ddagger}  			
\newcommand{\rstar}{\ast}  			

\providecommand{\mto}{\mapsto}

\renewcommand{\comp}{\,}
\newcommand{\komp}{\cdot}

\providecommand{\ins}{\oname{in}}
\providecommand{\out}{\oname{out}}
\providecommand{\tuo}{\oname{out}^\mone}

\newcommand{\nat}{\mathbb{N}}

\providecommand{\while}[2]{\mfix{while\mkern3mu{}}{#1}{\mkern1mu{}}{\mathbin{}#2}}
\providecommand{\dw}[2]{\mfix{do\mkern3mu{}}{\mathbin{}#1}{\mkern3mu{}while\mkern3mu{}}{\mathbin{}#2}}

\renewcommand{\ift}[3]{\mfix{if\mkern2mu{}}{\mathbin{}#1}{\mkern3mu{}then\mkern3mu{}}{\mathbin{} #2}{\mkern3mu else\mkern3mu{}}{#3}}

\newcommand{\Dec}{\BC_\BBT^{\mathsf d}}

\usepackage{savesym}

\savesymbol{emptyset}
\savesymbol{degree}
\savesymbol{leftmoon}
\savesymbol{rightmoon}
\savesymbol{fullmoon}
\savesymbol{newmoon}
\savesymbol{diameter}
\savesymbol{bigtimes}
\savesymbol{langle}
\savesymbol{rangle}
\savesymbol{triangleright}
\savesymbol{triangleleft}
\savesymbol{not}

\usepackage[matha,mathx]{mathabx}

\restoresymbol{other}{langle}
\restoresymbol{other}{rangle}
\restoresymbol{other}{triangleright}
\restoresymbol{other}{triangleleft}
\restoresymbol{other}{not}
\restoresymbol{other}{emptyset}

\providecommand{\mplus}{{\scriptscriptstyle\bf+}} 	       

\usepackage{xspace}

\newcommand{\axname}[1]{\text{\upshape\bfseries\textsf{#1}}\xspace}

\newcommand{\FIX}{\hyperref[it:fix]{\axname{\normalcolor Fixpoint}}\xspace}
\newcommand{\NAT}{\hyperref[it:nat]{\axname{\normalcolor Naturality}}\xspace}
\newcommand{\UNI}{\hyperref[it:uni]{\axname{\normalcolor Uniformity}}\xspace}
\newcommand{\COD}{\hyperref[it:cod]{\axname{\normalcolor Codiagonal}}\xspace}
\newcommand{\DIN}{\hyperref[rem:din]{\axname{\normalcolor Dinaturality}}\xspace}
\newcommand{\SQR}{\hyperref[it:sqr]{\axname{\normalcolor Squaring}}\xspace}
\newcommand{\SUN}{\hyperref[it:suni]{\axname{\normalcolor Uniformity$^\bigstar$}}\xspace}

\newcommand{\WFIX}{\hyperref[it:wfix]{\axname{\normalcolor W-Fix}}\xspace}
\newcommand{\WAND}{\hyperref[it:wand]{\axname{\normalcolor W-And}}\xspace}

\newcommand{\WUNI}{\hyperref[it:wuni]{\axname{\normalcolor W-Uni}}\xspace}

\newcommand{\DWFIX}{\hyperref[it:dwfix]{\axname{\normalcolor DW-Fix}}\xspace}
\newcommand{\DWAND}{\hyperref[it:dwand]{\axname{\normalcolor DW-And}}\xspace}
\newcommand{\DWOR}{ \hyperref[it:dwor]{\axname{\normalcolor  DW-Or}}\xspace}
\newcommand{\DWUNI}{\hyperref[it:dwuni]{\axname{\normalcolor DW-Uni}}\xspace}

\usepackage{cmll}

\newcommand{\boolnot}{\mathrel{\texttt{\upshape{\raisebox{0.5ex}{\texttildelow}}}}\kern-1.5pt}
\newcommand{\boolor}{\mathrel{\texttt{\upshape|\kern-2.5pt|}}}   
\newcommand{\booland}{\mathrel{\texttt{\upshape\&\kern-.5pt\&}}} 
\renewcommand{\tt}{\oname{tt}}
\newcommand{\ff}{\oname{ff}}

\let\doendproof\endproof
\renewcommand\endproof{~\hfill$\qed$\doendproof}

\usepackage{lineno}

\usepackage{etoolbox} 

\newcommand*\linenomathpatch[1]{%
  \cspreto{#1}{\linenomath}%
  \cspreto{#1*}{\linenomath}%
  \csappto{end#1}{\endlinenomath}%
  \csappto{end#1*}{\endlinenomath}%
}

\linenomathpatch{equation}
\linenomathpatch{gather}
\linenomathpatch{multline}
\linenomathpatch{align}
\linenomathpatch{alignat}
\linenomathpatch{flalign}

\usepackage{ifdraft} 
\ifdraft{
  \usepackage[draft]{showlabels} 
  
  \usepackage[layout=footnote,draft]{fixme}
}{
 \usepackage[layout=footnote,final]{fixme}
}

\FXRegisterAuthor{sg}{asg}{SG}	

\let\cedilla\c
\renewcommand{\c}{\colon}

\newcommand{\lrule}[3]{\textbf{#1}~~\frac{#2}{#3}}
\newcommand{\itref}[1]{{\upshape\ref{#1}}}

\let\div\undefined
\newcommand{\div}{\delta}

\usepackage{tikz-cd}

\tikzstyle{shiftarr} = [
        rounded corners,%
        to path={--([#1]\tikztostart.center)
                     -- ([#1]\tikztotarget.center) \tikztonodes
                     -- (\tikztotarget)},
]

\tikzset{
      commutative diagrams/.cd
    , arrow style=tikz
    , diagrams={>=stealth}
    , row sep=large
    , column sep = huge
}

\tikzset{
    cong/.style={draw=none,edge node={node [sloped, allow upside down, auto=false]{$\cong$}}}
  , iso/.style={draw=none,every to/.append style={edge node={node [sloped, allow upside down, auto=false]{$\cong$}}}}
}

\usepackage{tikz,pgf,pgfmath}
\usepackage{relsize}

\usetikzlibrary{shapes,arrows}
\usetikzlibrary{calc}
\usetikzlibrary{fit,decorations.pathreplacing,arrows.meta,positioning}
\usetikzlibrary{decorations.pathreplacing}

\usetikzlibrary{decorations.pathreplacing}
\tikzset{%
  show curve controls/.style={
    postaction={
      decoration={
        show path construction,
        curveto code={
          \draw [blue]
            (\tikzinputsegmentfirst) -- (\tikzinputsegmentsupporta)
            (\tikzinputsegmentlast) -- (\tikzinputsegmentsupportb);
          \fill [red, opacity=0.5]
            (\tikzinputsegmentsupporta) circle [radius=.5ex]
            (\tikzinputsegmentsupportb) circle [radius=.5ex];
        }
      },
      decorate
}}}

\pgfkeys{/pgf/.cd,
  east ports/.initial={.5},
  west ports/.initial={.5},
}

\makeatletter

\pgfdeclareshape{gprof}{

  \inheritsavedanchors[from=rectangle]
  \inheritanchorborder[from=rectangle]

  \foreach \p in {north, north west, north east, center,
      west, east, mid, mid west, mid east, base, base west,
      base east, south, south west, south east}
  {\inheritanchor[from=rectangle]{\p}}

  \backgroundpath{

    \northeast \pgf@xa=\pgf@x \pgf@ya=\pgf@y
    \southwest \pgf@xb=\pgf@x \pgf@yb=\pgf@y

    \ifgbar@left
      \pgfpathmoveto{\pgfpoint{\pgf@xa}{\pgf@ya}}
      \pgfpathlineto{\pgfpoint{\pgf@xa}{\pgf@yb}}
      \pgfpathlineto{\pgfpoint{0.5*\pgf@xa+0.5*\pgf@xb}{\pgf@yb}}
      \pgfpathlineto{\pgfpoint{0.5*\pgf@xa+0.5*\pgf@xb}{\pgf@ya}}
    \else
      \pgfpathmoveto{\pgfpoint{0.5*\pgf@xa+0.5*\pgf@xb}{\pgf@ya}}
      \pgfpathlineto{\pgfpoint{0.5*\pgf@xa+0.5*\pgf@xb}{\pgf@yb}}
      \pgfpathlineto{\pgfpoint{\pgf@xb}{\pgf@yb}}
      \pgfpathlineto{\pgfpoint{\pgf@xb}{\pgf@ya}}
    \fi
    \pgfpathclose
    \pgfusepath{clip}
    \pgfsetcornersarced{
      \pgfpoint{.25*min(\pgf@xa-\pgf@xb,\pgf@ya-\pgf@yb)}{
                .25*min(\pgf@xa-\pgf@xb,\pgf@ya-\pgf@yb)}
    }
    \pgfpathrectanglecorners{\southwest}{\northeast}
    \pgfusepath{fill}
  }
}

\pgfdeclareshape{gbox}{

  \inheritsavedanchors[from=rectangle]
  \inheritanchorborder[from=rectangle]

  \foreach \p in {north, north west, north east, center,
    west, east, mid, mid west, mid east, base, base west,
    base east, south, south west, south east}
  {\inheritanchor[from=rectangle]{\p}}

  \inheritbackgroundpath[from=rectangle]
  \inheritbeforebackgroundpath[from=rectangle]
  \inheritbehindforegroundpath[from=rectangle]
  \inheritforegroundpath[from=rectangle]
  \inheritbeforeforegroundpath[from=rectangle]%

  \savedmacro\getdboxparameters{%
    \pgfmathsetlengthmacro\pgf@right@top@offset{\pgfkeysvalueof{/tikz/right top offset}}%
    \addtosavedmacro\pgf@right@top@offset%
    \pgfmathsetlengthmacro\pgf@right@mid@offset{\pgfkeysvalueof{/tikz/right mid offset}}%
    \addtosavedmacro\pgf@right@mid@offset%
    \pgfmathsetlengthmacro\pgf@left@bot@offset{\pgfkeysvalueof{/tikz/left bot offset}}%
    \addtosavedmacro\pgf@left@bot@offset%
    \pgfmathsetlengthmacro\pgf@left@mid@offset{\pgfkeysvalueof{/tikz/left mid offset}}%
    \addtosavedmacro\pgf@left@mid@offset%
  }%

  \anchor{geast}{
    \getdboxparameters%
    \pgf@sh@lib@gbox@geastanchor{\pgf@right@top@offset}{\pgf@right@mid@offset}
  }

  \anchor{gwest}{
    \getdboxparameters%
    \pgf@sh@lib@gbox@gwestanchor{\pgf@left@bot@offset}{\pgf@left@mid@offset}
  }

  \anchor{uwest}{
    \getdboxparameters%
    \pgf@sh@lib@gbox@geastanchor{-\pgf@left@bot@offset}{-\pgf@left@mid@offset}
    \pgf@ya=\pgf@y
    \southwest
    \pgf@y=\pgf@ya
  }

  \anchor{ueast}{
    \getdboxparameters%
    \pgf@sh@lib@gbox@gwestanchor{-\pgf@right@top@offset}{-\pgf@right@mid@offset}
    \pgf@ya=\pgf@y
    \northeast
    \pgf@y=\pgf@ya
  }

  \backgroundpath{

    \def\pgf@left@bot@offset{\pgfkeysvalueof{/tikz/left bot offset}}%
    \def\pgf@right@top@offset{\pgfkeysvalueof{/tikz/right top offset}}%
    \def\pgf@left@mid@offset{\pgfkeysvalueof{/tikz/left mid offset}}%
    \def\pgf@right@mid@offset{\pgfkeysvalueof{/tikz/right mid offset}}%
    \def\pgf@bar@thickness{\pgfkeysvalueof{/tikz/bar thickness}}%

    \ifx\tikz@fillcolor\pgfutil@empty\relax\else
    \pgfpathrectanglecorners{\southwest}{\northeast}%
    \pgfusepath{fill}
    \fi

    \southwest \pgf@xa=\pgf@x \pgf@ya=\pgf@y
    \northeast \pgf@xb=\pgf@x \pgf@yb=\pgf@y

    \ifx\tikz@strokecolor\pgfutil@empty\pgfsetfillcolor{black}\else
    \pgfsetfillcolor{\tikz@strokecolor}
    \fi

    \ifhide@bars\relax\else

    \pgfmathparse{int(0.5*\pgf@yb-\pgf@right@top@offset - 0.5*\pgf@ya - \pgf@right@mid@offset)}
    \ifnum\pgfmathresult<0\relax\else

    \fi

    \southwest \pgf@xa=\pgf@x \pgf@ya=\pgf@y
    \northeast \pgf@xb=\pgf@x \pgf@yb=\pgf@y

    \pgfmathparse{int(0.5*\pgf@yb-\pgf@left@mid@offset - 0.5*\pgf@ya - \pgf@left@bot@offset)}
    \ifnum\pgfmathresult<0\relax\else

      \fi
    \fi

   \pgfpathrectanglecorners{\southwest}{\northeast}%
   \pgfusepath{stroke}

  }%
}%

\def\pgf@sh@lib@gbox@gwestanchor#1#2{%
  \southwest%
  \pgf@ya=\pgf@y
  \northeast
  \pgf@yb=.5\pgf@y
  \advance\pgf@yb by .5\pgf@ya
  \advance\pgf@ya by #1
  \advance\pgf@yb by #2
  \southwest%
  \pgf@y=.5\pgf@ya
  \advance\pgf@y by .5\pgf@yb
}

\def\pgf@sh@lib@gbox@geastanchor#1#2{
  \northeast
  \pgf@ya=\pgf@y
  \southwest
  \pgf@yb=.5\pgf@y
  \advance\pgf@yb by .5\pgf@ya
  \advance\pgf@ya by #1
  \advance\pgf@yb by  #2
  \northeast%
  \pgf@y=.5\pgf@ya
  \advance\pgf@y by .5\pgf@yb
}

\pgfdeclareshape{box}{

  \inheritsavedanchors[from=rectangle]
  \inheritanchorborder[from=rectangle]

  \foreach \p in {north, north west, north east, center,
    west, east, mid, mid west, mid east, base, base west,
    base east, south, south west, south east}
  {\inheritanchor[from=rectangle]{\p}}

  \inheritbackgroundpath[from=rectangle]
  \inheritbeforebackgroundpath[from=rectangle]
  \inheritbehindforegroundpath[from=rectangle]
  \inheritforegroundpath[from=rectangle]
  \inheritbeforeforegroundpath[from=rectangle]%
  \inheritbehindbackgroundpath[from=rectangle]

  \savedmacro\box@eastports{%
    \def\box@eastports{\pgfkeysvalueof{/pgf/east ports},0}
  }

  \savedmacro\box@westports{%
    \def\box@westports{\pgfkeysvalueof{/pgf/west ports},0}
  }

  \backgroundpath{
    \pgfpathrectanglecorners{\southwest}{\northeast}%
  }%

  \pgfutil@g@addto@macro\pgf@sh@s@box{%
    \pgfmathparse{dim(\box@eastports)}\pgfmathresult
    \let\portnum=\pgfmathresult

    \pgfmathloop%
    \ifnum\pgfmathcounter=\portnum\relax%
    \else%
      \pgfutil@ifundefined{pgf@anchor@box@o\pgfmathcounter}{%
        \expandafter\xdef\csname pgf@anchor@box@o\pgfmathcounter\endcsname{%
          \noexpand\pgf@sh@lib@box@eastanchor{\pgfmathcounter}%
        }%
      }{}%
    \repeatpgfmathloop%

    \pgfmathparse{dim({\box@westports})}\pgfmathresult
    \let\portnum=\pgfmathresult

    \pgfmathloop%
    \ifnum\pgfmathcounter=\portnum\relax%
    \else%
      \pgfutil@ifundefined{pgf@anchor@box@i\pgfmathcounter}{%
        \expandafter\xdef\csname pgf@anchor@box@i\pgfmathcounter\endcsname{%
          \noexpand\pgf@sh@lib@box@westanchor{\pgfmathcounter}%
        }%
      }{}%
    \repeatpgfmathloop%
  }
}

\def\pgf@sh@lib@box@eastanchor#1{%
  \pgfmathparse{dim(\box@eastports)}\pgfmathresult
  \let\portnum=\pgfmathresult
  \ifnum#1=\portnum\relax%
    \pgfpointorigin%
  \else
    \pgfmathparse{{\box@eastports}[#1-1]}%
    \let\z\pgfmathresult
    \southwest
    \pgf@ya=\pgf@y%
    \advance\pgf@ya by -\z\pgf@y
    \northeast
    \advance\pgf@ya by \z\pgf@y
    \pgf@y=\pgf@ya%
  \fi
}

\def\pgf@sh@lib@box@westanchor#1{%
  \pgfmathparse{dim(\box@westports)}\pgfmathresult
  \let\portnum=\pgfmathresult
  \ifnum#1>\portnum\relax%
    \pgfpointorigin%
  \else
    \pgfmathparse{{\box@westports}[#1-1]}%
    \let\z\pgfmathresult
    \northeast
    \pgf@ya=\z\pgf@y%
    \southwest
    \advance\pgf@y by -\z\pgf@y
    \advance\pgf@y by \pgf@ya
  \fi
}

\tikzset{
  left bot offset/.initial=.05cm,
  left mid offset/.initial=.0cm,
  right mid offset/.initial=-.05cm,
  right top offset/.initial=-.0cm,
  bar thickness/.initial=.15cm,
}

\tikzset{
  guarded box/.style={
    draw,
    thick,
    text width=2em,
    minimum height=1.5em,
    align=center,
    font=\itshape,
    gbox,
    left bot offset=.05cm,
    left mid offset=-.05cm,
    right mid offset=-.05cm,
    right top offset=.05cm,
    append after command={[#1]
      let \p1=(\tikzlastnode.gwest), 
          \p2=(\tikzlastnode.geast),
          \p3=(\tikzlastnode.east),
          \p4=(\tikzlastnode.north) in
      node[gbar, 
        left, 
        text width=0em,
        inner sep=0pt,
        minimum width/.expanded=\pgfkeysvalueof{/tikz/bar thickness},
        minimum height=\y4-\y3-\pgfkeysvalueof{/tikz/left bot offset}-\pgfkeysvalueof{/tikz/left mid offset},
        at=(\p1)] {} 
      node[gbar, 
        right,
        text width=0em, 
        inner sep=0pt, 
        minimum width/.expanded=\pgfkeysvalueof{/tikz/bar thickness},
        minimum height=\y4-\y3+\pgfkeysvalueof{/tikz/right top offset}+\pgfkeysvalueof{/tikz/right mid offset}, 
        at=(\p2)] {} 
    }
  },
  hide bars/.code=\hide@barstrue,
  guarded box/.default=0
}

\newif\ifhide@bars
\newif\ifgbar@left

\tikzset{
  gbar/.style={
    draw,
    gprof,
    minimum size=.35cm,
    left/.code=\gbar@lefttrue,
    right/.code=\gbar@leftfalse,
  },
  gprof/.default=0
}

\makeatother

\tikzset{
  unguarded box/.style={
    draw,
    thick,
    text width=2em,
    minimum height=1.5em,
    align=center,
    box,
  },
}

\tikzset{
  annotation/.style={
    font=\footnotesize,
    draw,
    solid,
    circle,
    very thin,
    inner sep=.3pt
  }
}

\tikzset{
  sigma line/.style={draw, smooth} 
}

\makeatletter
\newcommand\currentcoordinate{\the\tikz@lastxsaved,\the\tikz@lastysaved}
\makeatother

\pgfkeys{
  /tikz/left pin offset/.initial=.5cm,
  /tikz/right pin offset/.initial=.5cm,
}

\tikzset{
  left pins/.code args={#1}{
       \foreach \i in {#1}
          \draw let \p1 = \i in \i to ([xshift=-\pgfkeysvalueof{/tikz/left pin offset}]\x1,\y1);
  },
  right pins/.code args={#1}{
       \foreach \i in {#1}
          \draw let \p1 = \i in \i to ([xshift=\pgfkeysvalueof{/tikz/right pin offset}]\x1,\y1);
  }
}

\tikzset{
  lpinch/.initial=.5,
  rpinch/.initial={\pgfkeysvalueof{/tikz/lpinch}},
  xwobble/.initial=.5,
  sigma line/.style={
    to path={
      let \n1 = {\pgfkeysvalueof{/tikz/lpinch}}, 
          \n2 = {\pgfkeysvalueof{/tikz/rpinch}} in
      .. controls ($ (\tikztostart -| \tikztotarget) !\n1! (\tikztostart) $)
              and ($ (\tikztotarget -| \tikztostart) !\n2! (\tikztotarget) $) 
      .. (\tikztotarget) \tikztonodes
    }
  },
  ystretch/.initial=20,
  prat/.initial=.5,
  trace line/.style={
    to path={
      let \p1 = (\tikztostart),
          \p2 = (\tikztotarget),
          \n1 = {max(\y1,\y2)},
          \n2 = {\pgfkeysvalueof{/tikz/ystretch}},
          \n3 = {(\n1-\y1)},
          \n4 = {(\n1-\y2)},
          \n5 = {\pgfkeysvalueof{/tikz/prat}}, 
      in
         .. controls ++(5:\n2) and ++(-5:1.5*\n2) 
         .. ($(\x1,\n1) + (-.5*\n2,\n2)$) 
         .. controls ++(175:1*\n2) and ++(5:{.5*(\n2+\y1-\y2)}) 
         .. ($(\x2,\n1) + ({\n5*(\n2+\y1-\y2)}, \n2)$) 
         .. controls ++(5:{-1.5*\n2-.7*(\y1-\y2)}) and ++(-185:{\n2+.2*(\y1-\y2)}) 
         .. (\tikztotarget) \tikztonodes
    }
  },
  sigma line/.default=.5cm
}

\tikzset{
  cprop/.style={scale=0.8, inner ysep=0pt,every node/.style={scale=0.8,line width=.6pt}},
  guarded box/.append style={draw, line width=.6pt},
  unguarded box/.append style={draw, line width=.6pt}
}

\tikzset{
  bullet/.style={
      circle
    , fill
    , minimum size=4pt
    , inner sep=0pt
    , outer sep=0pt
  }
  , nodes=bullet
}

\usepackage{hypcap}
\title{Shades of Iteration: from Elgot to Kleene}
\author{Sergey Goncharov}
\institute{Friedrich-Alexander-Universit\"at Erlangen-N{\"u}rnberg, Germany}

\allowdisplaybreaks

\begin{document}

\maketitle

\begin{abstract}
Notions of iteration range from the arguably most general \emph{Elgot 
iteration} to a very specific \emph{Kleene iteration}. The fundamental nature
of Elgot iteration has been extensively explored by Bloom and Esik in the form
of \emph{iteration theories}, while \emph{Kleene iteration} became extremely 
popular as an integral part of (untyped) formalisms, such as automata theory, 
regular expressions and Kleene algebra. Here, we establish a formal connection 
between Elgot iteration and Kleene iteration in the form of Elgot monads and Kleene 
monads, respectively. We also introduce a novel class of \emph{while-monads}, 
which like Kleene monads admit a relatively simple description in algebraic 
terms. Like Elgot monads, while-monads cover a large variety of models that meaningfully 
support while-loops, but may fail the Kleene algebra laws, or even fail to support 
a Kleen iteration operator altogether.
\end{abstract}

\section{Introduction}\label{sec:intro}
Iteration is fundamental in many areas of computer science, such as semantics, 
verification, theorem proving, automata theory, formal languages, computability theory,
compiler optimisation, etc. An early effort to identifying a generic notion of iteration
is due to Elgot~\cite{Elgot75}, who proposed to consider an \emph{algebraic theory} induced
by a notion of abstract machine (motivated by \emph{Turing machines}, and their 
variants) and regard iteration as an operator over this algebraic theory.

Roughly speaking, an algebraic theory carries composable spaces of 
morphisms~$L(n,m)$, indexed by natural numbers $n$ and $m$ and including  
all functions from $n$ to $m$\footnote{Here we 
identify numbers $n\in\nat$ with finite ordinals $\{0,\ldots,n-1\}$.}, called \emph{base morphisms}. 
For example, following Elgot, one can consider as $L(n,m)$ the 
space of all functions $n\times S\to m\times S$ representing transitions from a 
machine state ranging over $n$ to a machine state ranging over~$m$, and updating the background store
over $S$ (e.g.\ with $S$ being the Turing machine's tape) in the meanwhile. In modern speech, $L(n,m)$ is 
essentially the space of Kleisli morphisms $n\to Tm$ of the \emph{state monad} $T=(\argument\times S)^S$.
Then a machine over $m$ halting states and $n$ non-halting states is represented 
by a morphism in $L(n,{m+n})$, and the iteration operator is meant to compute
a morphism in~$L(n,m)$, representing a run of the machine, obtained by feedbacking 
all non-halting states. This perspective has been extensively 
elaborated by Bloom and Esik~\cite{BloomEsik93} who identified the ultimate equational 
theory of Elgot iteration together with plenty other examples of algebraic theories
induced by existing semantic models, for which the theory turned out to be sound 
and complete. 

By replacing natural numbers with arbitrary objects of a category with finite
coproducts and by moving from purely equational to a closely related and
practically appealing quasi-equational theory of iteration, one arrives at
\emph{(complete) Elgot monads}~\cite{AdamekMiliusEtAl10,GoncharovSchroderEtAl18},
which are monads $T$, equipped with an iteration operator
\begin{align}\label{eq:elg-iter}\tag{$\istar$}
\lrule{}{f\c X\to T(Y+X)}{f^\istar\c X\to TY}
\end{align} 
In view of the connection between computational effects and monads, pioneered by 
Moggi~\cite{Moggi91}, Elgot monads provide arguably the most general model of 
iteration w.r.t.\ functions carrying computational effects, such as mutable store,
non-determinism, probability, exceptional and abnormal termination, input-output
actions of process algebra. The standard way of semantics via \emph{domain theory} 
yields a general (least) fixpoint operator, which sidelines Elgot iteration and 
overshadows its fundamental role. This role becomes material again when it comes to 
the cases when the standard scenario cannot be applied or is difficult to 
apply, e.g.\ in constructive setting~\cite{Goncharov21}, for deterministic hybrid system 
semantics~\cite{GoncharovJakobEtAl18}, and infinite trace semantics~\cite{LevyGoncharov19}.

In contrast to Elgot iteration, \emph{Kleene iteration}, manifested by \emph{Kleene 
algebra}, is rooted in logic and automata theory~\cite{Kleene56}, and crucially relies on 
non-determinism. The laws of Kleene algebra are from the outset determined by a rather 
conservative observation model, describing discrete events, coming one after another
in linear order and in finite quantities. Nevertheless, 
Kleene algebra and thus Kleene iteration proved to be extremely successful (especially 
after the celebrated complete algebraic axiomatization of Kleene algebra by Kozen~\cite{Kozen94}) 
and have been accommodated in various formalizations and verification frameworks from those 
for concurrency~\cite{HoareMollerEtAl09} to those for modelling hybrid 
systems~\cite{Platzer08}. A significant competitive advantage of Kleene iteration is 
that it needs no (even very rudimental) type grammar for governing well-definedness 
of syntactic constructs, although this cannot be avoided when extending Kleene 
algebra with standard programming features~\cite{Kozen02,Aboul-HosnKozen06,KozenMamouras13}.
Semantically, just as Elgot iteration, Kleene iteration can be reconciled with 
computational effects, leading to \emph{Kleene monads}~\cite{Goncharov10}, which 
postulate Kleene iteration with the type profile:
\begin{align}\label{eq:kle-iter}\tag{$\rstar$}
\lrule{}{f\c X\to TX}{f^\rstar\c X\to TX}
\end{align} 
Given $f$, $f^\rstar$ self-composes it non-deterministically indefinitely many times. 
In contrast to Elgot monads, the stock of computational effects modelled by 
Kleene monads is rather limited, which is due to the fact that many computational 
effects are subject to laws, which contradict the Kleene algebra laws. For a simple 
example, consider the computational effect of exception raising, constrained by the 
law, stating that postcomposing an exception raising program by another program
is ineffective. Together with the Kleene algebra laws, we obtain a havoc:
\begin{align*}
\oname{raise} e_1 = \oname{raise} e_1;\ \bot = \bot = \oname{raise} e_2;\ \bot = \oname{raise} e_2,
\end{align*}
where $\bot$ is the unit of non-deterministic choice.
This and similar issues led to a number of proposals to weaken Kleene 
algebra laws~\cite{FokkinkZantema94,Moller07,McIverRabehajaEtAl11,GomesMadeiraEtAl17} 
(potentially leading to other classes of monads, somewhere between Elgot and Kleene), although 
not attempting to identify the weakest set of such laws from the foundational
perspective. At the same time, it seems undebatable that Kleene iteration and 
the Kleene algebra laws yield the most restricted notion of iteration. 

We thus obtain a spectrum of potential notions of iteration between Elgot monads and Kleene monads.
The goal of the present work is, on the one hand to explore this spectrum, 
and on the other hand to contribute into closing the conceptual gap between Kleene iteration
and Elgot iteration. To that end, we introduce \emph{while-monads}, which capture
iteration in the conventional form of while-loops. Somewhat surprisingly, despite 
extensive work on axiomatizing iteration in terms of~\eqref{eq:elg-iter}, a corresponding
generic axiomatization in terms of ``while'' did not seem to be available.
We highlight the following main technical contributions of the present work:
\begin{itemize}
  \item We provide a novel axiomatization of Kleene algebra laws, which is 
  effective both for Kleene algebras and Kleene monads (\autoref{pro:klee-eq});
  \item We show that the existing axiomatization of Elgot monads is minimal (\autoref{pro:mini});
  \item We establish a connection between Elgot monads and while-monads (\autoref{thm:main});
  \item We render Kleene monads as Elgot monads with additional properties (\autoref{thm:elg-klee}).
\end{itemize}
%



\section{Preliminaries}\label{sec:prelim}
We rely on rudimentary notions and facts of category theory, as used in semantics, most
notably \emph{monads}~\cite{Awodey10}. 
For a (locally small) category $\BC$ we denote by~$|\BC|$ the class of its objects and 
by~$\BC(X,Y)$ the set of morphisms from $X\in |\BC|$ to~$Y\in|\BC|$. We often omit
indices at components of natural transformations to avoid clutter.~$\Set$ will 
denote the category of classical sets and functions, i.e.\ sets and functions 
formalized in a classical logic with the law of excluded middle (we will make no use of the axiom of choice).
By $\brks{f,g}\c X\to Y\times Z$ we will denote the \emph{pairing} of two morphisms $f\c X\to Y$
and $g\c X\to Z$ (in a category with binary products), and dually, by $[f,g]\c X+Y\to Z$
we will denote the \emph{copairing} of $f\c X\to Z$ and $g\c Y\to Z$ (in a category with
binary coproducts). By $\bang\c X\to 1$ we will denote terminal morphisms (if $1$ is 
an terminal object).  

An ($F$-)algebra for an endofunctor $F\c\BC\to\BC$ is a pair $(A,a\c FA\to A)$. 
Algebras form a category under the following notion of morphism: $f\c A\to B$
if a morphism from $(A,a)$ to $(B,b)$ if $b f = (Ff)\comp a$. The \emph{initial algebra}
is an initial object of this category (which may or may not exit). We denote this 
object~$(\mu F,\ins)$. ($F$-)coalgebras are defined dually as pairs of the form 
$(A\comma a\c {A\to FA})$. The final coalgebra will be denoted~$(\nu F\comma\out)$.
By Lambek's Lemma~\cite{Lambek68}, both $\ins$ and $\out$ are isomorphisms, and we 
commonly make use of their inverses $\ins^\mone$ and $\tuo$. 

\section{Monads for Computation}\label{sec:mona}
We work with monads represented by \emph{Kleisli triples} $(T,\eta,(\argument)^\klstar)$ where~$T$ is
a map~$|\BC|\to |\BC|$, $\eta$ is the family $(\eta_X\c X\to TX)_{X\in |\BC|}$ and 
$(\argument)^\klstar$ sends $f\c X\to TY$ to $f^\klstar\c TX\to TY$ in such a way that
the standard \emph{monad laws} 
\begin{align*} 
\eta^{\klstar}=\id, && 
f^{\klstar}\eta=f,  && 
(f^{\klstar} g)^{\klstar}=f^{\klstar}g^{\klstar} 
\end{align*}  
hold true. It is then provable that $T$ extends to a 
functor with $Tf = (\eta f)^\klstar$ and~$\eta$ to a \emph{unit} natural transformation.
Additionally, we can define the \emph{multiplication} natural transformation $\mu\c TT\to T$
with $\mu_X = \id^\klstar$
(thus extending $T$ to a monoid in the category of endufunctors).  
We preferably use bold letters, e.g.~$\BBT$, for monads, to contrast with 
the underlying functor $T$. The axioms of monads entail that the morphisms of the 
form $X\to TY$ determine a category, called \emph{Kleisli category}, and denoted $\BC_{\BBT}$,
under \emph{Kleisli composition} $f\komp g = f^\klstar\comp g$ with~$\eta$ as the 
identity morphism. Intuitively, Kleisli category is the category of (generalized)
effectful programs w.r.t.\ $\BC$ as the category of ``pure'', or effectless, programs.
More precisely, we will call \emph{pure} those morphisms in $\BC_{\BBT}$ that are of 
the form~$\eta\comp f$.
We thus use diagrammatic composition $g;f$ alongside and equivalently to functional 
composition $f\komp g$, as the former fits with the sequential composition operators of 
traditional programming languages. 

A monad $\BBT$ is \emph{strong} if it comes with 
a natural transformation $\tau_{X,Y}\c X\times TY\to T(X\times Y)$ called \emph{strength} 
and satisfying a number of coherence conditions~\cite{Moggi91}. Any monad on $\Set$
is canonically strong~\cite{Kock72}.

\begin{example}[Monads]\label{exa:mon}
Recall some computationally relevant monads on~$\Set$ 
(all monads on $\Set$ are strong \cite{Moggi91}).
\begin{enumerate}[wide]
\item\label{it:may} \emph{Maybe-monad}: $TX = X + 1$, $\eta(x) = \inl x$, $f^\klstar(\inl x) = f(x)$,
${f^\klstar(\inr\star) = \inr\star}$. 
\item\label{it:pow}\emph{Powerset monad}: $TX = \PSet X$, $\eta(x) = \{x\}$, $f^\klstar(S\subseteq X) = \{y\in f(x)\mid x\in S\}$.
\item\label{it:plo} $TX = \{S\mid S\in\PSet^\mplus (X + 1), \text{if $S$ is infinite then~} \inr\star\in S\}$ where $\PSet^\mplus$ is the non-empty powerset functor, 
  ${\eta(x) = \{\inl x\}}$, $f^\klstar(S\subseteq X) = \{y\in f(x)\mid{\inl x\in S}\}\cup(\{\inr\star\}\cap S)$.
\item\label{it:exc}\emph{Exception monad}: $TX=X+E$ where $E$ is a fixed (unstructured) non-empty 
set of exceptions, $\eta(x) = \inl x$, $f^\klstar(\inl x) = f(x)$, ${f^\klstar(\inr e) = \inr e}$.
\item\label{it:wri}\emph{Non-deterministic writer monad}: $TX = \PSet(M\times X)$ 
where~$(M\comma\eps\comma\bullet)$ is any monoid, $\eta(x) = \{(e,x)\}$,
$f^\klstar(S\subseteq M\times X) = \{(n\bullet m,y) \mid (m,x)\in S\comma
(n,y)\in f(x)\}$.
\item\label{it:dis}\emph{Discrete sub-distribution monad}: $TX = {\{d\c
  [0,1]\to X\mid\sum_{x\in X} d(x)\leq 1\}}$ (the \emph{supports} of $d$,
  $\{x\in X\mid d(x)>0\}$ are necessarily countable -- otherwise the sum~$\sum_{x\in X} d(x)$ would
  diverge), $\eta(x)$ is the \emph{Dirac distribution} $\delta_x$, centred
  in~$x$, i.e.\ $\delta_x(y) = 1$ if $x=y$, $\delta_x(y) = 0$ otherwise, $({f\c
  X\to\mathcal{D} Y)^\klstar}(d)(y)={\sum_{x\in X} f(x)(y)\cdot d(x)}$.
\item
\emph{Partial state monad:}\label{it:sta}
  $TX=(X\times S+1)^S$, where $S$ is a fixed set of global states, $\eta(x)(s) = \inl(x,s)$,
  $f^\klstar(g\c S\to X\times S+1)(s) = \inr\star$ if $g(s) = \inr\star$
  and $f^\klstar(g\c S\to Y\times S+1)(s) = f(x)(s')$ if $g(s) = \inl(x,s')$.
%
\item\emph{Partial interactive input:}\label{it:in} 
$TX=\nu\gamma.\,((X+\gamma^I)+1)$, where $I$ is a set of input values, $\eta(x) = \tuo(\inl\inl x)$,
$(f\c X\to TY)^\klstar$ is the unique such morphism $f^\klstar\c TX\to TY$ that (eliding 
the isomorphisms $T\iso (\argument+T^I)+1$)
\begin{align*}
f^\klstar(\inl\inl x) =\;& f(x),&
f^\klstar(\inl\inr h) =\;& \inl\inr (f^\klstar\comp h),&
f^\klstar(\inr\star) =\;& \inr\star.&
\end{align*}
Intuitively, $p\in TX$ is a computation that either finishes and gives a result 
in $X$, or takes an input from $I$ and continues recursively, or (unproductively) 
diverges.
\item\emph{Partial interactive output:}\label{it:out} 
$TX=\nu\gamma.\,((X+\gamma\times O)+1)$, where $O$ is a set of output values, 
$\eta(x) = \tuo(\inl\inl x)$, $(f\c X\to TY)^\klstar$ is the unique such morphism $f^\klstar\c TX\to TY$ that (eliding 
the isomorphisms $T\iso (\argument+T\times O)+1$)
\begin{align*}
f^\klstar(\inl\inl x) =\;& f(x),&
f^\klstar(\inl\inr (p,o)) =\;& \inl\inr (f^\klstar(p),o),&
f^\klstar(\inr\star) =\;& \inr\star.&
\end{align*}
The behaviour of $p\in TX$ is as in the previous case, except that it outputs 
to $O$ instead of expecting an input from $I$ in the relevant branch.
\end{enumerate}
Kleisli categories are often equivalent to categories with more familiar independent
descriptions. For example, the Kleisli category of the maybe-monad 
is equivalent to the category of partial functions and the Kleisli category of the 
powerset monad is equivalent to the category of relations. Under the monads-as-effects
metaphor, partial functions can thus be regarded as possibly non-terminating 
functions and relations as non-deterministic functions.

The above examples can often be combined. E.g.\ non-deterministic stateful
computations are obtained as $TA=S\to\PSet(A\times S)$.
The Java monad of~\cite{JacobsPoll03},
\begin{displaymath}
TX=S\to (X\times S + E\times S)+1
\end{displaymath}
with $S$ the set of states and $E$ the set of exceptions.
\end{example}


\section{Kleene Monads}\label{sec:klee} 

\begin{figure*}[t]
\begin{flalign*}
\intertext{\itshape Idempotent semiring laws:}\\[-2ex]
    \hspace{4em}&\textit{idempotence:}      &\qquad& f\lor f = f&\\                                     
    \quad&\textit{commutativity:}           && f\lor g = g\lor f &\\
    \quad&\textit{neutrality of $\bot$:}    && f\lor\bot = f &\\
    &\textit{associativity of $\lor$:}      && f\lor (g\lor h) = (f\lor g)\lor h &\\              
    &\textit{associativity of $;$~:}        && f; (g; h) = (f; g); h &                                                            
\\[2ex]
    \quad&\textit{right strictness:}        && f;\bot =\bot&\\                                     
    &\textit{right neutrality of~$\eta$:}   && f;\eta =f\\
    &\textit{right distributivity:}         && (f\lor g); h =f; h\lor g; h\hspace{7em}  
\\[2ex]
    &\textit{left strictness:}              && \bot; f = \bot\\
    &\textit{left neutrality of~$\eta$:}    && \eta; f = f\\
    &\textit{left distributivity:}          && f; (g\lor h)=f; g\lor f; h
\\[2ex]
\intertext{\itshape Iteration laws:}
    &\textit{right unfolding:}              && f^\rstar = \eta\lor f;f^\rstar\\
    &\textit{right induction:}              && f; g\leq f\implies f; g^\rstar\leq f
\\[2ex]
    &\textit{left unfolding:}               && f^\rstar = \eta\lor f^\rstar;f\\
    &\textit{left induction:}               && f; g\leq g\implies f^\rstar; g\leq g
\end{flalign*}
  \caption{Axioms Kleene algebras/monads.}
  \label{fig:klee-ax}
\end{figure*}
A Kleene algebra can be concisely defined as an idempotent semiring $(S\comma\bot\comma\eta\comma\lor\comma\,;\,)$ 
equipped with an operator $(\argument)^\rstar\c S\to S$,
such that 
\begin{itemize}
  \item $g; f^\rstar$ is the least (pre-)fixpoint of $g\lor (\argument); f$,
  \item $f^\rstar; h$ is the least fixpoint of $h\lor f;(\argument)$,
\end{itemize}
where the order is induced by $\lor$: $f\leq g$ if $f\lor g = g$. We assume 
here and henceforth that sequential composition $;$ binds stronger than $\lor$.

More concretely, a Kleene 
algebra $(S\comma\bot\comma\eta\comma\lor\comma\,;\,,(\argument)^\rstar)$ is an algebraic structure, 
satisfying the laws in~\autoref{fig:klee-ax}.
A categorical version of Kleene algebra emerges
as a class of monads, called~\emph{Kleene monads}~\cite{GoncharovSchroderEtAl09},
which can be used for interpreting effectful languages with iteration and 
non-determinism. 
%
\begin{definition}[Kleene-Kozen Category/Kleene Monad]
We say that a category $\BC$ is a \emph{Kleene-Kozen category} if $\BC$ is 
enriched over bounded (i.e.\ possessing a least element) join-semilattices
and strict join-preserving morphisms, and there is \emph{Kleene iteration} operator
\begin{displaymath}
  (\argument)^\rstar\c\BC(X,X)\to\BC(X,X),
\end{displaymath} 
such that, given $f\c Y\to Y$, $g\c Y\to Z$
and ${h\c X\to Y}$, $g\comp f^\rstar$ is the least (pre\dash)fixpoint of $g\lor (\argument)\comp f$
and $f^\rstar\comp h$ is the least (pre\dash)fixpoint of $h\lor f\comp(\argument)$. 

A monad $\BBT$ is a \emph{Kleene monad} if $\BC_{\BBT}$ is a Kleene-Kozen 
category.
\end{definition}
Recall that a monoid is nothing but a single-object category, whose morphisms 
are identified with monoid elements, and whose identity morphisms and morphism 
composition are identified with monoidal unit and composition. This suggests a 
connection between Kleene-Kozen categories and Kleene algebras.
\begin{proposition}\label{pro:monad-alg}
A Kleene algebra is precisely a Kleene-Kozen category with one object.
\end{proposition}
We record the following characterization of Kleene-Kozen categories 
(hence, also of Kleene algebras by~\autoref{pro:monad-alg}).
\begin{proposition}\label{pro:klee-eq}
A category $\BC$ is a Kleene-Kozen category iff 
\begin{itemize}[wide]
  \item $\BC$ is enriched over bounded join-semilattices and strict join-preserving morphisms;
  \item there is an operator $(\argument)^\rstar\c\BC(X,X)\to\BC(X,X)$,
  such that 
  \begin{enumerate}[leftmargin=10ex, labelindent=12ex]
    \item\label{it:klee-eq1} $f^\rstar = \id\lor f^\rstar\comp f$;
    \item\label{it:klee-eq2} $\id^\rstar = \id$;
    \item\label{it:klee-eq4} $f^\rstar = (f\lor\id)^\rstar$;    
    \item\label{it:klee-eq3} $h\comp f = f\comp g$ ~implies~ $h^\rstar\comp f = f\comp g^\rstar$.
  \end{enumerate} 
\end{itemize} 
\end{proposition}
\begin{proof}
Let us show \emph{necessity}. 
\begin{enumerate}[wide]
  \item The law $f^\rstar = \id\lor f^\rstar\comp f$ holds by assumption.
  \item Since $\id = \id\lor\id\comp\id$, $\id$ is a
fixpoint of $f\mto\id\lor f\comp\id$, and thus $\id^\rstar\leq\id$. Also 
$\id^\rstar = \id\lor\id^\rstar\comp\id\geq\id$. Hence $\id = \id^\rstar$ by 
mutual inequality.
   \item To show that $(f\lor \id)^\rstar\leq f^\rstar$, note that 
   $f^\rstar = \id\lor f^\rstar\comp f = \id\lor f^\rstar\comp f\lor f^\rstar =\id\lor f^\rstar\comp (f\lor\id)$,
   and use the fact that $(f\lor\id)^\rstar$ is the least fixpoint. The opposite 
   inequality is shown analogously, by exploiting the fact that $f^\rstar$ is a 
   least fixpoint.
   \item Suppose that $h\comp f = f\comp g$,
and show that $h^\rstar\comp f = f\comp g^\rstar$. Note that
\begin{displaymath}
f\lor (h^\rstar\comp f)\comp g = f\lor h^\rstar\comp f\comp g 
= f\lor h^\rstar\comp h\comp f 
= h^\rstar\comp f,
\end{displaymath}
%
i.e.\ $h^\rstar\comp f$ satisfies the fixpoint equation for $f\comp g^\rstar$,
and therefore $h^\rstar\comp f\geq f\comp g^\rstar$. By a symmetric argument, 
$h^\rstar\comp f\leq f\comp g^\rstar$, hence $h^\rstar\comp f= f\comp g^\rstar$.
\end{enumerate}
We proceed with \emph{sufficiency}. Suppose that $(\argument)^\rstar$ is as described 
in the second clause of the present proposition. Observe that by combining assumptions~\itref{it:klee-eq1} 
and~\itref{it:klee-eq3} we immediately obtain the dual version of~\itref{it:klee-eq1}, 
which is $f^\rstar = \id\lor f\comp f^\rstar$.
Now, fix $f\c Y\to TY$ and $g\c Y\to TZ$, and show that
$f^\rstar\comp g$ is the least fixpoint of $g\lor f\comp (\argument)$ -- we omit
proving the dual property, since it follows by a dual argument. From 
$f^\rstar = \id\lor f^\rstar\comp f$ we obtain $f^\rstar\comp g = g\lor f\comp(f^\rstar\comp g)$,
i.e.\ $f^\rstar\comp g$ is a fixpoint. We are left to show that it is the least
one. Suppose that $h = g\lor f\comp h$ for some $h$, which entails 
\begin{displaymath}
  (\id\lor f)\comp h = h\lor f\comp h = g\lor f\comp h\lor f\comp h = h = h\comp\id.
\end{displaymath}
Since $g\leq h$, using assumptions~\itref{it:klee-eq2},~\itref{it:klee-eq4}~and~\itref{it:klee-eq3}, we obtain
\begin{displaymath}
  f^\rstar\comp g\leq (\id\lor f)^\rstar\comp g \leq (\id\lor f)^\rstar\comp h = h\comp\id^\rstar = h,
\end{displaymath}
as desired.
\end{proof}
The axioms of Kleene monads do not in fact need a monad, and can be interpreted
in any category. We focus on Kleisli categories for two reasons: (i)~in practice, 
Kleene-Kozen categories are often realized as Kleisli categories, and 
monads provide a compositional mechanism for constructing more Kleene-Kozen
categories by generalities; (ii) we will relate Kleene monads and Elgot monads,
and the latter are defined by axioms, which do involve both general Kleisli morphisms 
and the morphisms of the base category. 
\begin{example}\label{exa:kle}
Let us revisit~\autoref{exa:mon}. Many monads therein fail to be Kleene simply 
because they fail to support binary non-determinism. \autoref{exa:mon}.\itref{it:dis} is an interesting 
case, since we can define the operation of probabilistic choice 
$+_p\c\BC(X,TY)\times\BC(X,TY)\to\BC(X,TY)$ indexed by $p\in [0,1]$, meaning that 
$x+_p y$ is resolved to $x$ with probability $p$ and to $y$ with probability $1-p$.
For every ${x\in X}$, $f(x)+_p g(x)$ is a \emph{convex sum} of the distributions~$f(x)$ and~$g(x)$. This 
operation satisfies the axioms of \emph{barycentric algebras} (or, \emph{abstract convex sets}~\cite{Stone49}), 
which are somewhat similar to those of a monoid, but with the multiplication operator indexed over 
$[0,1]$. To get rid of this indexing, one can remove the requirement that probabilities
sum up to at most $1$ and thus obtain spaces of \emph{valuations}~\cite{VaraccaWinskel06}
instead of probability distributions. Valuations can be conveniently added pointwise,
and thus defined addition satisfies monoidal laws, but fails idempotence, hence
still does not yield a Kleene monad. Given two valuations $v$
and $w$, we also can define $v\lor w$ as the pointwise maximum. This  
satisfies the axioms of semilattices, but fails both distributivity laws.

\autoref{exa:mon}.\itref{it:plo} is the $\Set$-reduct 
or \emph{Plotkin powerdomain}~\cite{Plotkin76} over a flat domain. It supports 
proper non-deterministic choice, but the only candidate for $\bot$ is not a unit 
for it. 

Kleene monads of~\autoref{exa:mon} are only~\itref{it:pow} and~\itref{it:wri}. 
The non-deterministic state monad over $TX=(\PSet(X\times S))^S$ obtained by 
adapting~\autoref{exa:mon}.\itref{it:sta} in the obvious way is also~Kleene.
\end{example}
Except for the powerset monad, our examples of Kleene monads are in fact obtained 
by generic patterns.
\begin{proposition}\label{pro:klee}
Let $\BBT$ be a Kleene monad. Then so are 
\begin{enumerate}[wide]
  \item the state monad transformer $(T(\argument\times S))^S$ for every $S$;
  \item the writer monad transformer $T(M\times\argument)$ for every monoid $(M,\eps,\bullet)$ 
  if $\BBT$ is strong and strength $\tau_{X,Y}\c X\times TX\to T(X\times Y)$ respects the 
  Kleene monad structure, as follows:
  \begin{align}\label{eq:str-klee}
\begin{gathered}
    \tau\comp (\id\times\bot) = \bot,\qquad
    \tau\comp (\id\times f^\rstar) = (\tau\comp (\id\times f))^\rstar,\\[1ex]
    \tau\comp (\id\times (f\lor g)) = \tau\comp (\id\times f)\lor \tau\comp (\id\times g).
  \end{gathered}
  \end{align}
\end{enumerate}
\end{proposition}
\begin{proof}
\begin{enumerate}[wide]
  \item By definition, the Kleisli category of the state transformer is equivalent
  to the full subcategory of $\BC_{\BBT}$ over the objects of the form $X\times S$
  (using the isomorphism $\BC(X,(T(Y\times S))^S)\iso \BC(X\times S,T(Y\times S))$). 
  The enrichment, the iteration operator and the axioms are clearly restricted along 
  the induced inclusion functor.
  \item
  The semilattice structure for every $\BC(X,T(M\times Y))$ is inherited from~$\BC_{\BBT}$, 
  but to show enrichment, the strictness and the distributivity laws must be 
  verified manually.
  For every $f\c X\to T(M\times Y)$, let 
  $f^\circ\c M\times X\to T(M\times Y)$ be as follows
  \begin{align*}
  M\times X\xto{\tau\comp (\id\times f)} T(M\times (M\times Y))\,\iso\, T((M\times M)\times Y) \xto{T(\bullet\times\id)} T(M\times Y).
  \end{align*}
  The assumptions~\eqref{eq:str-klee} entail the following identities:
  \begin{flalign*}
  &&\bot^\circ &=\bot\Label{l1}& (f\lor g)^\circ &= f^\circ\lor g^\circ \Label{l2}& (f\brks{\eps,\id})^\circ &= f \Label{l3}\\
  &&&&&& (f^\circ\komp g)^\circ &= f^\circ\komp g^\circ \Label{l4}& f^\circ\brks{\eps,\id} &= f \Label{l5}
\end{flalign*}
  Kleisli composition of the transformed monad sends $f\c X\to T(M\times Y)$
  and $g\c Y\to T(M\times Z)$ to the Kleisli composition~$g^\circ\komp f$ of~$\BBT$.
  Left strictness and right distributivity are then obvious, while right strictness and
  left distributivity follow too by~\eqref{l1},\eqref{l2}: $\bot^\circ\komp f=\bot\komp f=\bot$, 
  $(f\lor g)^\circ\komp h=(f^\circ\lor g^\circ)\komp h = f^\circ\komp h\lor g^\circ\komp h$.
  Kleene star for the transformed monad is defined as $(f^\circ)^\rstar\brks{\eps,\id}$ for 
  every $f\c X\to T(M\times X)$ where $\eps\c 1\to M$ is the monoid unit.
  \begin{flalign*}
  &&(f^{\circ})^\rstar\brks{\eps,\id} 
    =&\; (\eta\lor f^{\circ} \komp (f^{\circ})^\rstar) \brks{\eps,\id}\\
  &&  =&\; (\eta\lor f^{\circ} \komp (f^{\circ})^\rstar)\komp\eta\brks{\eps,\id}\\
  &&  =&\; \eta\brks{\eps,\id}\lor f^{\circ} \komp (f^{\circ})^\rstar\komp\eta\brks{\eps,\id}\\*
  &&  =&\; \eta\brks{\eps,\id}\lor f^{\circ} \komp (f^{\circ})^\rstar\brks{\eps,\id},\\[2ex]
  &&(f^{\circ})^\rstar\brks{\eps,\id} 
    =&\; (\eta\lor (f^{\circ})^\rstar\komp f^{\circ} ) \brks{\eps,\id}\\* 
  &&  =&\; (\eta\lor (f^{\circ})^\rstar\komp f^{\circ} )\komp\eta\brks{\eps,\id}\\ 
  &&  =&\; \eta\brks{\eps,\id}\lor (f^{\circ})^\rstar\komp f^{\circ} \komp\eta\brks{\eps,\id}\\
  &&  =&\; \eta\brks{\eps,\id}\lor (f^{\circ})^\rstar\komp f&\by{\eqref{l5}}\\
  &&  =&\; \eta\brks{\eps,\id}\lor ((f^{\circ})^\rstar\brks{\eps,\id})^\circ\komp f.&\by{\eqref{l3}}
  \end{flalign*}
If $f^\circ\komp g\leq g$ then by~\eqref{l3}, $((f^{\circ})^\rstar\brks{\eps,\id})^\circ\komp g=(f^\circ)^\rstar\komp g\leq g$.
Analogously, if $f^\circ\komp g\leq f$ then by~\eqref{l4} $f^\circ\komp g^\circ = (f^\circ\komp g)^\circ\leq f^\circ$,
and therefore $f^\circ\komp (g^\circ)^\rstar \brks{e,\id} \leq  f^\circ\brks{e,\id}= f$.
Thus, the axioms of iteration are all satisfied.
\qed
\end{enumerate}
\noqed\end{proof}
\begin{example}
Note that the powerset monad $\PSet$ is a Kleene monad with~$f^\rstar$
calculated as a least fixpoint of $\eta\lor(\argument)\komp f$. 
\begin{enumerate}[wide]
  \item By~\autoref{pro:klee}.1, $(\PSet(\argument\times S))^S$ is a Kleene monad.
  \item By applying \autoref{pro:klee}.2, to the free monoid $A^\star$ of finite strings over an alphabet $A$
  we obtain that $\PSet(A^\star\times\argument)$.
\end{enumerate}
\begin{figure*}[t]
\FIX 
\vspace{-2ex}
\begin{center}
  \includegraphics[scale=1]{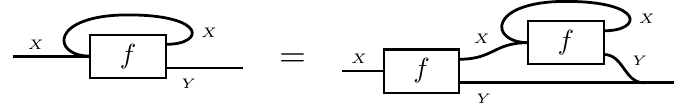}
\end{center}
\NAT 
\vspace{-2ex}
\begin{center}
  \includegraphics[scale=1]{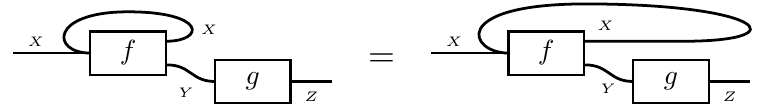}
\end{center}
\bigskip
\COD 
\vspace{-2ex}
\begin{center}
  \includegraphics[scale=1]{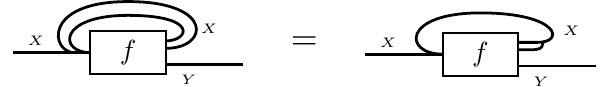}
\end{center}
\bigskip
\UNI 
\vspace{-2ex}
\begin{center}
  \includegraphics[scale=1]{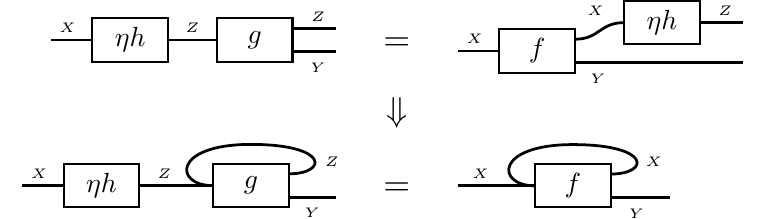}
\end{center}
  \caption{Axioms of Elgot monads.}
  \label{fig:elg-ax}
\end{figure*}
It is easy to see that for every Kleene monad $\BBT$, $\Hom(1,T1)$ is a Kleene 
algebra. By applying this to the above clauses we obtain correspondingly the standard 
\emph{relational} and \emph{language-theoretic} models of Kleene algebra~\cite{KozenSmith96}.
\end{example}
\section{Elgot Monads}\label{sec:elg}
A general approach to monad-based iteration is provided by Elgot monads. 
We continue under the assumption that $\BC$ supports finite coproducts. This, 
in particular, yields an if-the-else operator sending $b\in\BC(X,X+X)$ and 
$f,g\in\BC(X,Y)$ to $\ift{b}{p}{q} = [q,p]\komp b\in\BC(X,Y)$. Note that for any 
monad $\BBT$ on $\BC$, $\BC_{\BBT}$ inherits finite coproducts.
\begin{definition}[Elgot monad]\label{def:elgot}
  An \emph{Elgot monad} in a category with binary coproducts is a monad $\BBT$ equipped with an
  \emph{Elgot iteration} operator 
  \begin{displaymath}
  (\argument)^\istar\c\BC(X,T(Y+X))\to\BC(X,TY), 
  \end{displaymath}
  subject to the following principles:
\begin{flalign}
&\label{it:fix}\tag*{}\axname{Fixpoint}:       [\eta, f^{\istar}] \komp f = f^{\istar} & (f\c X\to T(Y+X))\\[1ex]
&\label{it:nat}\tag*{}\axname{Naturality}:     g\komp f^{\istar} = ([\eta\inl\komp g, \eta\inr] \komp f)^{\istar} & (g\c Y \to TZ, f\c X\to T(Y+X))\\[1ex]
&\label{it:cod}\tag*{}\axname{Codiagonal}:     f^{\istar\istar} = ([\eta,\eta\inr] \komp f)^{\istar} & (f\c X \to T((Y + X) + X))\\
&\label{it:uni}\tag*{}\axname{Uniformity}:     \lrule{}{g \komp \eta h = \eta(\id+h)\komp f}{g^{\istar} \komp \eta h = f^{\istar}} & ({h\c X\to Z}, g\c Z\to T(Y + Z),\\[-1ex]&& f\c X \to T(Y + X))\notag
\end{flalign}
\end{definition}

\noindent
These laws are easier to grasp by depicting them graphically (\autoref{fig:elg-ax}),
more precisely speaking, as \emph{string diagrams}~(cf.~\cite{JoyalStreetEtAl96,Hasegawa03}
for a rigorous treatment in terms of monoidal categories).
Iterating $f$ is depicted as a feedback loop. It is 
then easy to see that while~\FIX expresses the basic fixpoint property of iteration,
\NAT and~\COD are essentially rearrangements of wires. The~\UNI law is a form of 
induction: the premise states that $\eta h$ can be pushed over $g$, so that at the same 
time~$g$ is replaced by $f$, and the conclusion is essentially the result of 
closing this transformation under iteration. \UNI is therefore the only law, which alludes 
to pure morphisms. Intuitively, the morphisms $f$ and $g$
can be seen as programs operating correspondingly on $X$ and~$Y$ as their state
spaces, and $h\c X\to Y$ is a map between these state spaces. \UNI thus ensures that 
the behaviour of iteration does not depend on the shape of the state
space.
It is critical 
for this view that~$h$ is pure, i.e.\ does not trigger any side-effects.
\begin{remark}[Divergence]\label{rem:div}
Every Elgot monad comes together with the definable (unproductive) divergence 
constant $\div = (\eta\inr)^\istar$. Graphically, $\div\c X\to T\iobj$ will be depicted as 
\makebox
{
\begin{tikzpicture} [cprop]
  \draw [sigma line, thick] (0,0) to ++(6.5ex,0) node[circle,draw,fill=white,inner sep=0pt,minimum size=4pt] {};
\end{tikzpicture} 
}, symmetrically to the depiction of the initial morphism $\bang\c\iobj\to TX$ as
\makebox
{
\begin{tikzpicture} [cprop]
  \draw [sigma line, thick] (0,0) to ++(-6.5ex,0) node[bullet] {};
\end{tikzpicture} 
}.
\end{remark}
\begin{example}[Elgot Monads]\label{exa:elg}
Clauses~\itref{it:may}--\itref{it:out} of~\autoref{exa:mon} all define
Elgot monads. A standard way of introducing Elgot iteration is enriching the Kleisli category 
over pointed complete partial orders and defining $(f\c X\to T(Y+X))^\istar$ as
a least fixpoint of the map $[\eta,\argument]\komp f\c\BC(X,TY)\to\BC(X,TY)$ by
the \emph{Kleene fixpoint theorem}. This scenario covers~\itref{it:may}--\itref{it:sta}.
In all these cases, we inherit complete partial order structures on $\Set(X,TY)$
by extending canonical complete partial order structures from $TY$ pointwise.
In particular, in~\itref{it:exc}, we need to chose the divergence element $\div\in E$. This
choice induces a \emph{flat domain} structure on $X+E$: $x\appr y$ if $x=y$ or $x=\div$.
The induced divergence constant in the sense of \autoref{rem:div} then coincides with 
$\div$, and hence there are at least as many distinct Elgot monad structures on the 
exception monad as exceptions.

Clauses~\itref{it:in} and~\itref{it:out} fit a different pattern. 
For every Elgot monad~$\BBT$ and every endofunctor $H$, if all final coalgebras 
$T_HX = \nu\gamma.\,T(X+H\gamma)$ exist then $T_H$ extends to an Elgot monad~\cite{GoncharovJakobEtAl18},
called the \emph{coalgebraic generalized resumption transform} of $\BBT$. 
This yields~\itref{it:in} and~\itref{it:out} by taking $\BBT$ to be the maybe-monad 
in both cases and $HX=X^I$ and $HX=O\times X$ respectively.
\end{example}
\begin{remark}[Dinaturality]\label{rem:din}
A classical law of iteration, which is not included in~\autoref{def:elgot}, 
is the \DIN law, which has the following graphical representation: 
\begin{center}
  \includegraphics[scale=1]{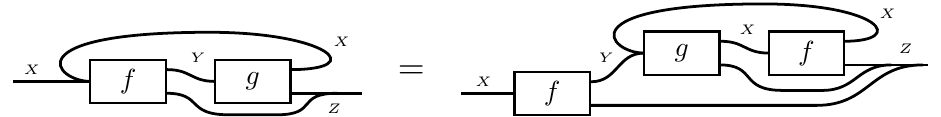}
\end{center}
This law 
has been used in one of the equivalent axiomatization
of iteration theories~\cite{BloomEsik93} (under the name \emph{``composition identity''})
and thus was initially inherited in the definition of Elgot monads~\cite{AdamekMiliusEtAl10,GoncharovSchroderEtAl18}. 
However, \'Esik and Goncharov~\cite{EsikGoncharov16} latter discovered that 
\DIN is derivable in presence of \UNI. 
\end{remark}
\autoref{rem:din} poses the question, if the present axiomatization of Elgot 
monads possibly contains further derivable laws. Here, we resolve it in the negative.

\begin{proposition}\label{pro:mini}
The axiomatization in~\autoref{def:elgot} is minimal.
\end{proposition}
\begin{proof}
For every axiom, we construct a separating example that fails that axiom, but 
satisfies the other three. Every example is a monad on $\Set$.
\begin{itemize}[wide]
  \item\FIX: For any monad $\BBT$, equipped with a natural transformation 
  $p\c 1\to TX$, we can define $f^\istar = [\eta,p\bang]\komp f$ for a
  given $f\c X\to {T(Y+X)}$. It is easy to see that~\NAT,~\COD and~\UNI are satisfied,
  but~\FIX need not to, e.g.\ with $\BBT$ being the non-deterministic writer monad
  (\autoref{exa:mon}.\ref{it:wri}) over the additive monoid of natural numbers~$\nat$.
  \item\NAT: Let $\BBT=\PSet$ and let $f^\istar(x) = Y$
  for every $f\c X\to T(Y+X)$ and every $x\in X$. Note that every $f\c X\to T(Y+Z)$
  is equivalent to a pair $(f_1\c X\to TY,f_2\c X\to TZ)$ and $[g,h]\komp f= g\komp f_1\cup h\komp f_2$
  for any $g\c Y\to TV$, $h\c Z\to TV$. This helps one to see that all the axioms, except \NAT hold true, e.g.\ 
  $([\eta,f^\istar]\komp f)(x) = f_1(x)\cup Y = Y = f^\istar(x)$.
  \NAT fails, because $g\komp\div = g\komp (\eta\inr)^\istar \neq ([\eta\inl\komp g, \eta\inr] \komp \eta\inr)^{\istar} = \div$,
  since the image of $\div\c X\to TY$ is~$\{Y\}$, while the image of $g\komp\div$,
  aka the image of $g$, need not be~$\{Y\}$. 
  \item\COD: Consider the exception monad transform $TX = {\PSet(2^\star\times X\cup 2^\omega)}$
  of the non-deterministic writer monad 
  over the free monoid~$2^\star$. This is canonically an Elgot monad, and let us 
  denote by $(\argument)^\iistar$ the corresponding iteration operator. Every
  $f\c X\to T(Y+X)$, using the isomorphism $T(Y+X)\iso\PSet(2^\star\times Y\cup 2^\omega)\times\PSet(2^\star\times X)$,
  induces a map $\hat f\c X\to \PSet(2^\star\times X)$. Let $f^\istar\c X\to TY$ 
  be as follows:~$f^\istar(x)$ is the union of $f^\iistar(x)$ and the set 
  \begin{align*}
  \{w\in 2^\omega\mid \exists u\in 2^\star.\,uw = w_1w_2\ldots,(w_1,x_1)\in\hat f(x), (w_2,x_2)\in\hat f(x_1),\ldots\}.
  \end{align*}
  That $(\argument)^\istar$ satisfies \FIX, \NAT
  and \UNI follows essentially from the fact that so does $(\argument)^\iistar$. To show 
  that $(\argument)^\istar$ fails \COD, consider $g\c 1\to \PSet((2^\star+2^\star)\cup 2^\omega)$,
  with $g(\star) = \{\inl 0,\inr 1\}$. Let $f$ be the composition of~$g$ with 
  the obvious isomorphism $\PSet((2^\star+2^\star)\cup 2^\omega)\iso T((0+1)+1)$.
  Now $([\eta,\eta\inr]\komp f)^\istar(\star) = 2^\omega\neq \{0^\omega,1^\omega\} 
  = f^{\istar\istar}(\star)$.
  \item\UNI: Consider the exception monad on $TX = X + \{0,1\}$. This can be
  made into an Elgot monad in two ways: by regarding either $0$ or $1$ as the divergence
  element. Given $f\c X\to T(Y+X)$, we let $f^\istar$ be computed as a least fixpoint 
  according to the first 
  choice if $X$ is a singleton and according to the second choice otherwise. The 
  axioms except \UNI are clearly satisfied. To show that \UNI fails, let $|X|>2$,
  $|Z|=1$, $g=\eta\inr$, $f=\eta\inr$, $h=\bang$. The premise of \UNI is thus satisfied, 
  while the conclusion is not, since $f^\istar$ is constantly $1$ and $g^\istar$
  is constantly $0$.
  \qed
  \end{itemize}
\noqed\end{proof}
Although we cannot lift any of the Elgot monad laws, \NAT can be significantly restricted. 
\begin{proposition}\label{pro:nat}
In the definition of Elgot monad, \NAT can be equivalently replaced by its instance
with $g$ of the form $\eta\inr\c Y\to T(Y'+Y)$.
\end{proposition}
\begin{proof}
We use the fact that \DIN and the following law, called \axname{Squaring} are derivable
from \FIX, \COD and \UNI~\cite[Lemma~31]{GoncharovSchroder18}:

\medskip
\begin{center}\label{it:sqr}
  \includegraphics[scale=1]{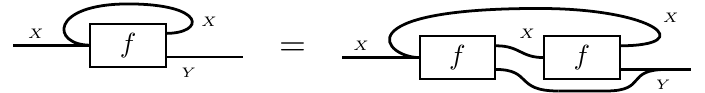}
\end{center}
Let us fix $g\c Y \to TZ$, $f\c X\to T(Y+X)$ and proceed to show that
$g\komp f^{\istar} = ([\eta\inl\komp g, \eta\inr] \komp f)^{\istar}$.
Let $w = [\eta\inl\komp g,\eta\inr\komp f]\c Y+X\to T(Z+(Y+X))$, and note that
\FIX entails 
\begin{align}\label{eq:nat1}
w^\istar\komp\eta\inl =&\; g
\end{align}
The goal will follow from the identities
\begin{align}
w^\istar\komp\eta\inr =&\; g\komp f^{\istar}\label{eq:nat2}\\
w^\istar\komp\eta\inr =&\; ([\eta\inl\komp g, \eta\inr] \komp f)^{\istar}\label{eq:nat3}
\end{align}
Let us show~\eqref{eq:nat2}, using an allowed instance of \NAT:
\begin{flalign*}
&&w^\istar\komp\eta\inr
  =\;& ([\eta,\eta\inr]\komp [\eta\inl\inl\komp g,T(\inr\inl+\inr)\comp f])^\istar\komp\eta\inr\\
&&=\;& ([\eta\inl\inl\komp g,T(\inr\inl+\inr)\comp f])^{\istar\istar}\komp\eta\inr&\by{\COD}\\
&&=\;& [\eta, w^\istar]\komp ([\eta\inl\inl\komp g,T(\inr\inl+\inr)\comp f])^{\istar}\komp\eta\inr&\by{\FIX}\\
&&=\;& [\eta, w^\istar]\komp (T(\inr\inl+\id)\comp f)^{\istar}\komp&\by{\UNI}\\
&&=\;& [\eta, w^\istar]\komp \eta\inr\inl\komp f^\istar&\by{\NAT}\\
&&=\;& w^\istar\komp\eta\inl\komp f^\istar\\
&&=\;& g\komp f^\istar.&\by{\eqref{eq:nat1}}
\intertext{Finally, let us show~\eqref{eq:nat3}:}
&&w^\istar\komp\inr
  =\;& [\eta,w^\istar]\komp w\komp\inr&\by{\FIX}\\
&&=\;& [\eta,w^\istar]\komp [\eta\inl\komp g,\eta\inr\komp f]\komp\inr\\
&&=\;& w^\istar\komp f\\
&&=\;& ([\eta\inl,w]\komp w)^\istar\komp f&\by{\SQR}\\
&&=\;& ([\eta\inl,[\eta\inl\komp g,\eta\inr\komp f]]\komp [\eta\inl\komp g,\eta\inr\komp f])^\istar\komp f&\\
&&=\;& [\eta,([\eta\inl,\eta\inr\komp f]\\&&&\qquad\komp [\eta\inl\komp g, [\eta\inl\komp g,\eta\inr]\komp f])^\istar]\komp\eta\inr\komp f&\\
&&=\;& ([\eta\inl,[\eta\inl\komp g, [\eta\inl\komp g,\eta\inr]\komp f]]\komp \eta\inr\komp f)^\istar&\by{\DIN}\\
&&=\;& ([\eta\inl\komp g, [\eta\inl\komp g,\eta\inr]\komp f]\komp f)^\istar&\\
&&=\;& ([\eta\inl, [\eta\inl\komp g,\eta\inr]\komp f]\komp [\eta\inl\komp g,\eta\inr]\komp f)^\istar&\\
&&=\;& ([\eta\inl\komp g,\eta\inr]\komp f)^\istar.&\by{\SQR}
\end{flalign*}
\end{proof}

\section{While-Monads}\label{sec:while}
We proceed to develop a novel alternative characterization of Elgot monads in more 
conventional for computer science terms of while-loops.
%
%
%
\begin{definition}[Decisions]
Given a monad $\BBT$ on~$\BC$, we call any family 
$(\Dec(X)\subseteq\BC(X,T(X+X)))_{X\in|\BC|}$, a family of \emph{decisions}
if every $\Dec(X)$ contains~$\eta\inl$, $\eta\inr$, and is closed 
under~$\oname{if}$-$\oname{then}$-$\oname{else}$.
\end{definition}
We encode logical operations on decisions as follows:
\begin{flalign*}
  &&\ff &\,= \eta\inl, & b\booland c =\,& \ift{b}{c}{\ff}, & \boolnot b =\,& \ift{b}{\ff}{\tt},\qquad\\
  &&\tt &\,= \eta\inr, & b\boolor c  =\,& \ift{b}{\tt}{c}.
\end{flalign*}
By definition, decisions can range from the smallest family with $\Dec(X) = 
\{\ff,\tt\}$, to the greatest one with $\Dec(X)=\BC(X,T(X+X))$.

\begin{remark}
Our notion of decision is maximally simple and general. An 
alternative are morphisms of the form $b\c X\to T2$, from which we can
obtain $X\xto{\brks{\id,b}} X\times T2\xto{\tau} T(X\times 2)\iso T(X+X)$ if $\BBT$
is strong, with~$\tau$ being the strength. The resulting 
decision $d$ would satisfy many properties we are not assuming generally, e.g.\
$\ift{d}{\tt}{\ff} = \eta$. Both morphisms 
of the form $X\to T2$ and $X\to T(X+X)$ are relevant in semantics as decision making
abstractions -- this is explained in detail from the perspective of categorical logic
by Jacobs~\cite{Jacobs16}, who uses the names \emph{predicates} and \emph{instruments}
correspondingly (alluding to physical, in particular, quantum experiments).
\end{remark}
Elgot monads are essentially the semantic gadgets for effectful while-languages. In fact, 
we can introduce a semantic $\oname{while}$-operator and express it via Elgot iteration.
Given $b\in\Dec(X)$ and $p\in\BC(X,TX)$, let 
\begin{align}\label{eq:while-as-it}
\while{b}{p} = (\ift{b}{p;\tt}{\ff})^\istar,
\end{align}
or diagrammatically, $\while{b}{p}$ is expressed as
\begin{center}
  \includegraphics[scale=1]{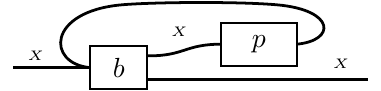}
\end{center}
%
%
%
It is much less obvious that, conversely, Elgot iteration can be defined via~$\oname{while}$,
and moreover that the entire class of Elgot monads can be rebased on~$\oname{while}$.
We dub the corresponding class of monad \emph{while-monads}.
\begin{definition}[While-Monad]
A \emph{while-monad} is a monad $\BBT$, equipped with an operator 
\begin{align*}
\oname{while}\c\Dec(X)\times\BC(X,TX)\to\BC(X,TX),
\end{align*}
such that the following axioms are satisfied
\begin{flalign}
\hspace{-5em}  &&\axname{W-Fix}    &\hspace{1em}\while{b}{p} = \ift{b}{p; (\while{b}{p})}{\eta}\tag*{}\label{it:wfix}\\*[1ex]
\hspace{-5em}  &&\axname{W-Or}     &\hspace{1em}\while{(b\boolor c)}{p} = (\while{b}{p});\while{c}{(p;\while{b}{p})}\tag*{}\label{it:wor}\\[2ex]
\hspace{-5em}  &&\axname{W-And}    &\hspace{.5em}\lrule{}{\eta h;b = \eta u;\ff}{
  \while{(b\booland (c\boolor\eta u;\ff))}{p} = \while{b}{(\ift{c}{p}{\eta h})}}\tag*{}\label{it:wand}\\*[1.5ex]
\hspace{-5em}  &&\axname{W-Uni}    &\hspace{.5em}\lrule{}{\eta h; b = \ift{c}{\eta h';\tt}{\eta u;\ff}\qquad \eta h';p = q;\eta h}
  {\eta h;\while{b}{p} = (\while{c}{q});\eta u}\tag*{}\label{it:wuni}
\end{flalign}
\end{definition}
The laws of while-monads roughly correspond to~\FIX, \COD, \NAT  and~\UNI. This 
correspondence is somewhat allusive for~\WAND,
which under $u=\id$ instantiates to the nicer looking
\begin{displaymath}
  \lrule{}{\eta h;b = \ff}{
  \while{(b\booland c)}{p} = \while{b}{(\ift{c}{p}{\eta h})}}
\end{displaymath}
However, this instance generally seems to be insufficient. Let us still consider 
it in more detail.
The while-loop $\while{(b\booland c)}{p}$ repeats $p$ as long as both~$b$ and $c$
are satisfied, and $\while{b}{(\ift{c}{p}{\eta h})}$ repeats $(\ift{c}{p}{\eta h})$
as long as $b$ is satisfied, but the latter program still checks $c$ before running $p$ 
and triggers $\eta h$ only if $c$ fails. The equality in the conclusion of the rule is thus due to 
the premise, which ensures that once $\eta h$ is triggered, the loop is exited 
at the beginning of the next iteration. 

Note that using the following equations 
\begin{align}
\dw{p}{b}       =&\; p;\while{b}{p}\label{eq:do-wh}\\*
\while{b}{p}    =&\; \ift{b}{(\dw{p}{b})}{\eta}\label{eq:wh-do}
\end{align}
we can define $(\dw{p}{b})$ from $(\while{b}{p})$ and conversely obtain the latter 
from the former. Unsurprisingly, while-monads can be equivalently defined in terms 
of~$\oname{do-while}$.
\begin{lemma}\label{lem:dowhile}
Giving a while-monad structure on $\BBT$ is equivalent to equipping~$\BBT$ with 
an operator, sending every $b\in\Dec(X)$ and every $p\in\BC(X,TX)$ to 
$(\dw{p}{b})\in\BC(X,TX)$, such that the following principles hold true: 
\begin{flalign}
\hspace{-2em}&&\axname{DW-Fix}  &\hspace{1em}\dw{p}{b} = p;\ift{b}{(\dw{p}{b})}{\eta}\tag*{}\label{it:dwfix}&\\[1ex]
\hspace{-2em}&&\axname{DW-Or}   &\hspace{1em}\dw{p}{(b\boolor c)} = \dw{(\dw{p}{b})}{c}\tag*{}\label{it:dwor}\\[2ex]
\hspace{-2em}&&\axname{DW-And}  &\hspace{.5em}\lrule{}{\eta h;b = \eta u;\ff}{
  \begin{array}{rl}\ift{c}{\dw{p&}{(b\booland (c\boolor\eta u;\ff))}}{\eta u}\\[1ex] &\quad = \dw{(\ift{c}{p}{\eta h})}{b}\end{array}
}\tag*{}\label{it:dwand}\\[1.5ex]
\hspace{-2em}&&\axname{DW-Uni}  &\hspace{.5em}\lrule{}{
  \eta h;p = q;\eta h'\qquad\eta h'; b = \ift{c}{\eta h;\tt}{\eta u;\ff} 
}{
  \eta h;\dw{p}{b} = (\dw{q}{c});\eta u
}\tag*{}\label{it:dwuni}
\end{flalign}
The relevant equivalence is witnessed by the equations~\eqref{eq:do-wh} and~\eqref{eq:wh-do}.
\end{lemma}
Finally, we can prove the equivalence of while-monads and Elgot monads, under an 
expressivity assumption, stating that sets of decisions $\Dec$ are sufficiently 
non-trivial. Such an assumption is clearly necessary, for, as we indicated above,
the smallest family of decisions is the one with $\Dec(X) = 
\{\ff,\tt\}$, and it is not enough to express meaningful while-loops.
\begin{theorem}\label{thm:main}
Suppose that for all $X,Y\in |\BC|$, $\eta(\inl+\inr)\in\Dec(X+Y)$. Then~$\BBT$ is 
and Elgot monad iff it is a while-monad w.r.t.~$\Dec$. The equivalence is witnessed 
by mutual translations:~\eqref{eq:while-as-it} and
\begin{align}\label{eq:it-as-while}
f^\istar = \eta\inr; (\while{\eta(\inl+\inr)}{[\eta\inl,f]}); [\eta,\div].
\end{align}
Diagrammatically, $f^\istar$ is expressed as
\begin{center}
  \includegraphics[scale=1]{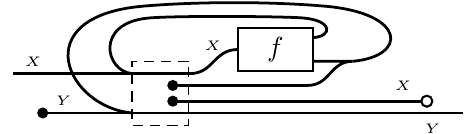}
\end{center}
for $f\c X\to T(Y+X)$.
\end{theorem}
\begin{proof}Note that~\eqref{eq:it-as-while} is equivalent to 
\begin{align}\label{eq:it-as-dowhile}
f^\istar = \eta\inr; \bigl(\dw{[\eta\inl,f]}{\eta(\inl+\inr)}\bigr); [\eta,\div].
\end{align}
First, we show that the indicated translations are mutually inverse.
\medskip\noindent
(i)~  $(\argument)^\istar\to\oname{while}\to (\argument)^\istar$: We need to show 
  that 
  \begin{displaymath}
    \eta\inr; (\ift{(\inl+\inr)}{[\eta\inl,f]; \eta\inr}{\eta\inl})^\istar; [\eta,\div] = f^\istar.
  \end{displaymath}
Indeed,
\begin{flalign*}
&& \eta\inr;&\; (\ift{(\inl+\inr)}{[\eta\inl,f]; \eta\inr}{\eta\inl})^\istar; [\eta,\div] \\*
&& &\;= [\eta,\div]\komp ([\eta\inl,\eta\inr\komp[\eta\inl,f]]\comp (\inl+\inr))^\istar\komp\eta\inr &\\
&& &\;= [\eta,\div]\komp [\eta\inl\inl,\eta\inr\komp f]^\istar\komp\eta\inr &\by{\NAT}\\
&& &\;= ([\eta\inl\komp[\eta,\div],\eta\inr]\komp [\eta\inl\inl,\eta\inr\komp f])^\istar\komp\eta\inr &\\
&& &\;= ([\eta\inl,\eta\inr\komp f])^\istar\komp\eta\inr &\\
&& &\;= ([\eta\inl,\eta\inr\komp f]\komp[\eta\inl,\eta\inr])^\istar\komp\eta\inr &\\
&& &\;= [\eta, ([\eta\inl,[\eta\inl,\eta\inr]]\komp \eta\inr\komp f)^\istar]\komp [\eta\inl,\eta\inr]\komp\eta\inr&\by{\DIN}\\
&& &\;= f^\istar.
\end{flalign*}
\medskip\noindent
(ii)~  $\oname{while}\to(\argument)^\istar\to\oname{while}$: We need to show that
  \begin{align}\label{eq:wh-it-wh}
    \eta\inr; \bigl(\while{\eta(\inl+\inr)}{[\eta\inl,\ift{b}{p;\tt}{\ff}]}\bigr); [\eta,\div] = \while{b}{p}.
  \end{align}
Observe that 
\begin{align*}
[\eta\inl,\ift{&b}{p;\tt}{\ff}] \\
=&\; [[\ff,\ff],[\ff,p;\tt]]\comp [\eta\inl;\ff, \ift{b}{\eta\inr;\tt}{\eta\inl;\ff}]\\
=&\; \ift{[\eta\inl;\ff, \ift{b}{\eta\inr;\tt}{\eta\inl;\ff}]}{[\ff,p;\tt]}{[\ff,\ff]},\\[2ex]
[\ff,\ff];&\, \eta(\inl+\inr)= \eta [\inl,\inl];\ff,
\intertext{and}
\eta(\inl+&\inr)\booland ([\eta\inl;\ff, \ift{b}{\eta\inr;\tt}{\eta\inl;\ff}]\boolor\eta [\inl,\inl];\ff)\\
=&\;[\eta\inl;\ff,\eta\inr;\tt]\booland [\eta\inl;\ff, \ift{b}{\eta\inr;\tt}{\eta\inl;\ff}]\\
=&\;[\eta\inl;\ff, \ift{b}{\eta\inr;\tt}{\eta\inl;\ff}]
\intertext{Hence, using~\WAND,}
\while{\eta&(\inl+\inr)}{[\eta\inl,\ift{b}{p;\tt}{\ff}]}\\
=\;& \while{[\eta\inl;\ff, \ift{b}{\eta\inr;\tt}{\eta\inl;\ff}]}{[\ff,p;\tt]}.
\intertext{Now,}
\eta\inr; [&\eta\inl;\ff, \ift{b}{\eta\inr;\tt}{\eta\inl;\ff}]\\
=\;& \ift{b}{\eta\inr;\tt}{\eta\inl;\ff},\\*[2ex]
\eta\inr; [&\ff,p;\tt] = p;\eta\inr.
\intertext{Hence, using~\WUNI,}
\eta\inr;(\while{[&\eta\inl;\ff, \ift{b}{\eta\inr;\tt}{\eta\inl;\ff}]}{[\ff,p;\tt]}); [\eta,\div]\\*
=\;&(\while{b}{p});\eta\inl; [\eta,\div]\\*
=\;&\while{b}{p},
\end{align*}
and we are done with the proof of~\eqref{eq:wh-it-wh}.

\medskip\noindent
(iii)~  Let us check that the laws of Elgot monads follow from the laws of while-monads
under the encoding~\eqref{eq:it-as-while}.

\medskip\noindent
\FIX:~ Suppose $f\c X\to T(Y+X)$ and let $p = \while{\eta(\inl+\inr)}{[\eta\inl,f]}$. 
Note that
\begin{align}\label{eq:fp-wh}
\eta\inl;p = \eta\inl,
\end{align}
which follows by applying~\WUNI to the identities $\eta\inl; \eta(\inl+\inr) = \ift{\ff}{\eta\inl;\tt}{\eta\inl;\ff}$
and $\eta\inl;[\eta\inl,f] = \eta\inl$ and noting that $\while{\ff}{(\eta\inl)} = \eta$.
Then
\begin{flalign*}
&& f^\istar &\;= \eta\inr; p; [\eta,\div] &\\
&&  &\;= \eta\inr; (\ift{(\inl+\inr)}{[\eta\inl,f]; p}{\eta}); [\eta,\div]&\by{\WFIX}\\
&&  &\;= [\eta,\div]\komp [\eta,p\komp [\eta\inl,f]] (\inl+\inr)\komp\eta\inr\\
&&  &\;= [\eta,\div]\komp p\komp f\\
&&  &\;= [\eta,\div]\komp [p\komp\eta\inl,p\komp\eta\inr]\komp f\\
&&  &\;= [\eta,\div]\komp [\eta\inl, p\komp \eta\inr]\komp f&\by{\eqref{eq:fp-wh}}\\
&&  &\;= [\eta, [\eta,\div]\komp p\komp \eta\inr]\komp f\\
&&  &\;= f;[\eta,\eta\inr; p; [\eta,\div]]\\
&&  &\;= [\eta,f^\istar] \komp f.
\end{flalign*}

\medskip\noindent
\NAT:~ By \autoref{pro:nat}, it suffices to show that 
\begin{align*}
  \eta h\komp f^{\istar} = (\eta(h+\id)\komp f)^{\istar} 
\end{align*}
for all $f\c X\to T(Y+X)$ and $h\c Y\to TZ$. Note that
\begin{align*}
\eta (h+\id); \eta(\inl+\inr) =&\; \ift{\eta(\inl+\inr)}{\eta (h+\id);\tt}{\eta (h+\id);\ff},\\
\eta (h+\id); [\eta\inl,\eta(h+\id) \komp f] =&\; [\eta\inl,f]; \eta (h+\id),
\end{align*}
hence, by \WUNI,
\begin{align*}
\eta (h+\id);&\while{\eta(\inl+\inr)}{[\eta\inl,\eta(h+\id) \komp f]}\\
 =\, (\!&\while{\eta(\inl+\inr)}{[\eta\inl,f]}); \eta (h+\id).
\end{align*}
Therefore, 
Then
\begin{flalign*}
&& \eta h\komp f^\istar 
    &\;= \eta\inr; (\while{\eta(\inl+\inr)}{[\eta\inl,f]}); [\eta,\div]; \eta h &\\*
&&  &\;= \eta\inr; (\while{\eta(\inl+\inr)}{[\eta\inl,f]}); \eta (h+\id); [\eta,\div]&\\
&&  &\;= \eta\inr;\eta (h+\id);(\while{\eta(\inl+\inr)}{[\eta\inl,\eta(h+\id) \komp f]}); [\eta,\div]\\
&&  &\;= (\eta(h+\id)\komp f)^{\istar},
\end{flalign*}
and we are done.

\medskip\noindent
\COD:~ Let $f\c X\to T((Y+X)+X)$. We will work with the translation~\eqref{eq:it-as-dowhile}. Let $Z=(Y+X)+X$, 
fix the following morphisms:
\begin{align*}
 p =&\; [[\eta\inl\inl,f],f]\c Z\to TZ,\\
 b =&\; \eta(\inl+\inr)\c Z\to T(Z+ Z),\\
 c =&\; [\eta(\inl\inl+\inl\inr),\eta\inr\inr]\c Z\to T(Z+Z),\\
 d =&\; [\eta,\eta\inr]\c Z\to T(Y+X),
\end{align*}
observe that $b\boolor c = c$ and that
\begin{align}
d;[\eta\inl,f;d] =&\; p;d,\notag\\
d;\eta(\inl+\inr) =&\; \ift{c}{d;\tt}{d;\ff},\label{eq:cod-dw0}
\end{align}
which by~\DWUNI entails
\begin{align}\label{eq:cod-dw1}
d; \dw{[\eta\inl,f;d]}{\eta(\inl+\inr)} = (\dw{p}{c}); d.
\end{align}
%
%
It is easy to see by~\DWUNI that 
\begin{align}\label{eq:cod-dw2}
f^\istar=\eta\inr; (\dw{[\eta\inl,f]}{\eta(\inl+\inr)}); d.
\end{align}%
Next, observe that $p=[\eta(\inl+\id),\eta\inr];[\eta\inl,f]$, and hence 
\begin{flalign*}
\quad (\dw{p}{b}); d
=&\; [\eta(\inl+\id),\eta\inr];(\dw{[\eta\inl,f]}{b}); d          &\by{\DWUNI}\\
=&\; [[\eta\inl\inl,\eta\inr],\eta\inr];(\dw{[\eta\inl,f]}{b}); d\\
=&\; [[\eta\inl,f^\istar],f^\istar]&\by{\eqref{eq:cod-dw2}}\\
=&\; d; [\eta\inl,f^\istar],
\end{flalign*}
which, together with~\eqref{eq:cod-dw0} by~\DWUNI yields
\begin{align}\label{eq:cod-dw4}
d; (\dw{[\eta\inl,f^\istar]}{\eta(\inl+\inr)}) = (\dw{(\dw{p}{b})}{c});  d
\end{align}
Finally, we can proceed with the proof of~\COD:
\begin{flalign*}
&&([\eta,\eta\inr]&\komp f)^\istar\\ 
&&  =&\; \eta\inr; (\dw{[\eta\inl,f\komp d]}{\eta(\inl+\inr)}); [\eta,\div]\\
&&  =&\; \eta\inr; d; (\dw{[\eta\inl,d\komp f]}{\eta(\inl+\inr)}); [\eta,\div]\\
&&  =&\; \eta\inr;(\dw{p}{c}); d; [\eta,\div]&\by{\eqref{eq:cod-dw1}}\\
&&  =&\; \eta\inr;(\dw{p}{(b\boolor c)}); d; [\eta,\div]\\
&&  =&\; \eta\inr;(\dw{(\dw{p}{b})}{c});  d; [\eta,\div]&\by{\DWOR}\\
&&  =&\; \eta\inr;d; (\dw{[\eta\inl,f^\istar]}{\eta(\inl+\inr)}); [\eta,\div]&\by{\eqref{eq:cod-dw4}}\\
&&  =&\; \eta\inr; (\dw{[\eta\inl,f^\istar]}{\eta(\inl+\inr)}); [\eta,\div]\\*
&&  =&\; f^{\istar\istar}.
\end{flalign*}

\medskip\noindent
\UNI:~ Suppose that $g \komp \eta h = \eta(\id+h)\komp \comp f$. This entails
\begin{align*}
\eta(\id+h);[\eta\inl,g]
=&\; [\eta\inl,g\komp \eta h]\\
=&\; [\eta\inl, [\eta\inl,\eta\inr\komp\eta h]\komp \comp f]\\
=&\; [\eta\inl,f];\eta(\id+h), 
\intertext{and note that}
\eta(\id+h);\eta(\inl+&\inr)\\
=&\; \eta(\inl+\inr h)\\
=&\; \eta(\inl+\inr);\eta((\id+h)+(\id+h))\\
=&\; \ift{(\inl+\inr)}{\eta(h+\id);\tt}{\eta(h+\id);\ff}.&&
\intertext{Therefore, using~\WUNI,}
\eta h; g^\istar 
=&\; \eta h; \eta\inr; (\while{\eta(\inl+\inr)}{[\eta\inl,g]}); [\eta,\div]\\*
=&\; \eta\inr h; (\while{\eta(\inl+\inr)}{[\eta\inl,g]}); [\eta,\div]\\
=&\; \eta\inr; \eta(\id+h); (\while{\eta(\inl+\inr)}{[\eta\inl,g]}); [\eta,\div]\\
=&\; \eta\inr; (\while{\eta(\inl+\inr)}{[\eta\inl,f]}); [\eta,\div]\\
=&\; f^\istar.
\end{align*}
\noindent
(iv)~ Finally, we check that the laws of while-monads follow from those of Elgot monads.
To that end, we verify the properties, listed in~\autoref{lem:dowhile}, which is 
equivalent. First of all, note that by~\eqref{eq:do-wh} and~\eqref{eq:while-as-it},
\begin{flalign*}
&&\dw{p}{b} =&\; p;\while{b}{p}                         &\by{\eqref{eq:do-wh}}\\
&&=&\; p;(\ift{b}{p;\ff}{\ff})^\istar                   &\by{\eqref{eq:while-as-it}}\\
&&=&\; ([\eta\inl,\eta\inr\komp p]\komp b)^\istar\komp p\\
&&=&\; [\eta, ([\eta\inl,\eta\inr\komp p]\komp b)^\istar]\komp \eta\inr\komp p\\
&&=&\; ([\eta\inl,b]\komp \eta\inr\komp p)^\istar&\by{\DIN}\\*
&&=&\; (b\komp p)^\istar.
\end{flalign*}
We will use the resulting encoding expansion of $\dw{p}{b}$ throughout. 
\begin{flalign*}
\intertext{\DWFIX:}
  \quad\dw{p}{b}
  =&\; (b\komp p)^\istar\\ 
  =&\; [\eta,(b\komp p)^\istar]\komp b\komp p&\by{\FIX}\\
  =&\; [\eta,\dw{p}{b}]\komp b\komp p\\ 
  =&\; p;\ift{b}{(\dw{p}{b})}{\eta}.
\intertext{\DWOR:} 
  \dw{p}{(b&\boolor c)}\\* 
  =&\; ((b\boolor c)\komp p)^\istar\\
  =&\; ([c,\eta\inr]\komp b\komp p)^{\istar}\\ 
  =&\; ([\eta,\eta\inr]\komp[\eta\inl\komp c,\eta\inr]\komp b\komp p)^{\istar}\\ 
  =&\; ([\eta\inl\komp c,\eta\inr]\komp b\komp p)^{\istar\istar}&\by{\COD}\\
  =&\; (c\komp (b\komp p)^\istar)^\istar&\by{\NAT}\\
  =&\; \dw{(\dw{p}{b})}{c}.
\intertext{\DWAND: Assuming that $\eta h;b = \eta u;\ff$,} 
\ift{c}{\dw{&p}{(b\booland (c\boolor\eta u;\ff))}}{\eta u}\\
=&\; [\eta u, ((b\booland (c\boolor\eta u;\ff))\komp p)^\istar]\komp c\\
=&\; [\eta u, ([\eta\inl,T(u+\id)\komp c]\komp b\komp p)^\istar]\komp c\\
=&\;  [\eta, ([\eta\inl, T(u+\id)\komp c]\komp b\komp p)^\istar]\komp T(u+\id)\komp c\\ 
=&\;  ([\eta\inl,b\komp p]\komp T(u+\id)\komp c)^\istar&\by{\DIN}\\
=&\;  ([\eta\inl\comp u,b\komp p]\komp c)^\istar\\ 
=&\;  ([b\komp\eta h,b\komp p]\komp c)^\istar&\by{assumption}\\
=&\;  (b\komp [\eta h,p]\komp c)^\istar\\ 
=&\;  (b\komp (\ift{c}{p}{\eta h}))^\istar\\ 
=&\; \dw{(\ift{c}{p}{\eta h})}{b}.
\intertext{\DWUNI: Assuming that $\eta h;p = q;\eta h'$ and $\eta h'; b = \ift{c}{\eta h;\tt}{\eta u;\ff}$, } 
\eta h;\dw{p&}{b}\\
=&\; (b\komp p)^\istar\komp\eta h \\
=&\; ([\eta\inl u,\eta\inr]\komp c\komp q)^\istar&\by{\UNI}\\
=&\; \eta\inl u\komp(c\komp q)^\istar&\by{\NAT}\\
=&\; (\dw{q}{c});\eta u.  
\end{flalign*}
This concludes the proof.
\end{proof}

\section{Kleene Monads as Elgot Monads}\label{sec:kleene-elgot}
If hom-sets of the Kleisli category of a while-monad $\BBT$ are equipped with a 
semilattice structure and every $\Dec(X)$ is closed under that structure, we can 
define Kleene iteration as follows:
\begin{align*}
p^\rstar = \while{(\ff\lor\tt)}{p}.
\end{align*}
That is, at each iteration we non-deterministically decide to finish or to continue.

Given a decision $b\in\Dec(X)$, let $b? = (\ift{b}{\eta}{\div})\in\BC(X,TX)$.
The standard way to express while-loops via Kleene iteration is as follows:
\begin{align*}
\while{b}{p} = (b?; p)^\rstar; (\boolnot b)? 
\end{align*}
If the composite translation $\oname{while}\to (\argument)^\rstar\to\oname{while}$
was a provable identity, this would essentially mean equivalence of Kleene iteration
and $\oname{while}$ with non-determinism. This is generally not true,
unless we postulate more properties that connect $\oname{while}$ and nondeterminism.
We leave for future work the problem of establishing a minimal set of such laws.
Here, we only establish the equivalence for the case when the induced Kleene iteration 
satisfies Kleene monad laws.
To start off, we note an alternative to \UNI, obtained by replacing 
the reference to pure morphisms with the reference to a larger class consisting of those~$h$, for which 
$\div\komp h = \div$. We need this preparatory step to relate Elgot iteration and Kleene iteration, since the 
latter does not hinge on a postulated class of pure morphisms, while the former
does. 
\begin{definition}[Strong Uniformity]
Given an Elgot monad $\BBT$, the \emph{strong uniformity} law is as follows: 
\begin{flalign}
\label{it:suni}\tag*{}\axname{Uniformity$^\bigstar$}:\qquad   \lrule{}{\div\komp h =\div \qquad g \komp  h = [\eta\inl,\eta\inr\komp h]\komp f}{g^{\istar} \komp h = f^{\istar}} &&
\end{flalign}
where $h\c X\to TZ$, $g\c Z\to T(Y + Z)$, and $f\c X \to T(Y + X)$.
\end{definition}
Clearly, \UNI is an instance of \SUN.
\begin{example}\label{exa:suni-trans}
An example of Elgot monad that fails \SUN can be constructed as 
follows. Let $\BBS$ be the reader monad transform of the maybe-monad on $\Set$: 
$SX = (X+1)^2$, which is an Elgot monad, since the maybe-monad is so and Elgotness 
is preserved by the reader monad transformer. Let $TX=X\times (X+1) + 1$ and note that $T$ is a retract of 
$S$ under 
\begin{align*}
\rho\c (X+1)\times (X+1)\iso X\times (X+1) + (X+1)\xto{\id+\bang} X\times (X+1) + 1.
\end{align*}
%
It is easy to check that $\rho$ is a congruence w.r.t.\ the Elgot monad structure, 
and it thus induces an Elgot monad structure on $\BBT$~\cite[Theorem~20]{GoncharovSchroderEtAl17}.

Now, let $T_{E} = T(\argument+E)$ for some non-empty $E$. The Elgot monad structure of 
$\BBT$ induces an Elgot monad structure on $\BBT_E$. However, $\BBT_E$
fails \SUN. Indeed, let $h\c X\to (X+E)\times ((X+E)+1) + 1$ and 
$f\c X\to ((1+X)+E)\times (((1+X)+E)+1) + 1$ be as follows:
\begin{align*}
h(x) = \inl(\inl x,\inl\inr e)&&
f(x) = \inl (\inr e, \inl\inl\inr x)
\end{align*}
where $e\in E$. Then $h\komp\div=\div$, $f^\istar(x) = \inl (\inr e, \inr\star)$, and 
$f\komp h = [\eta\inl,\eta\inr\komp h]\komp f$, but $(f^\istar\komp h)(x)=\inl(\inr e,\inl\inr e)\neq \inl (\inr e, \inr\star) = f^\istar(x)$. 
\end{example}
\begin{remark}
\autoref{exa:suni-trans} indicates that it is hard to come up with a
general and robust notion of Elgot iteration, which would confine to a single 
category, without referring to another category of ``well-behaved'' (e.g.\ pure)
morphisms. While the class of Elgot monads is closed under various 
monad transformers, the example shows that Elgot monads with strong uniformity 
are not even closed under the exception monad transformer.
\end{remark}
We are in a position to relate Kleene monads and Elgot monads.
\begin{theorem}\label{thm:elg-klee}
A monad $\BBT$ is a Kleene monad iff 
\begin{enumerate}
  \item $\BBT$ is an Elgot monad;
  \item the Kleisli category of $\BBT$ is enriched over join-semilattices 
  (without least elements) and join-preserving morphisms;
  \item $\BBT$ satisfies $(\eta\inl\lor\eta\inr)^\istar = \eta$;
  \item $\BBT$ satisfies \SUN.
\end{enumerate}
\end{theorem}
To prove the theorem, we need to mutually encode Kleene iteration and Elgot 
iteration. These encodings go back to C{\u{a}}z{\v{a}}nescu and \cedilla{S}tef{\v{a}}nescu~\cite{CazanescuStefanescu94}.
Some preparatory steps are needed.
The following is a standard property of Kleene algebra, which 
carries over to Kleene monads straightforwardly.
\begin{lemma}\label{lem:star-sum}
$(f\lor g)^\rstar = f^\rstar\komp (g\komp f^\rstar)^\rstar$.
\end{lemma}
Next, observe the following.
\begin{lemma}\label{lem:klee_lem}
For any monad $\BBT$, whose Kleisli category is enriched over join-semilattices 
and join-preserving morphisms, $[f_1,g_1]\lor [f_2,g_2] = [f_1\lor f_2,g_1\lor g_2]$
where $f_1,f_2\c X\to TZ$, $g_1,g_2\c Y\to TZ$.
\end{lemma}
\begin{proof}
The goal is entailed by the equations 
\begin{flalign*}
([f_1,g_1]\lor [f_2,g_2])\comp\inl =&\; [f_1\lor f_2,g_1\lor g_2]\comp\inl,\\
([f_1,g_1]\lor [f_2,g_2])\comp\inr =&\; [f_1\lor f_2,g_1\lor g_2]\comp\inr,
\intertext{
of which, we prove the first one. Indeed,} 
([f_1,g_1]\lor [f_2,g_2])\comp\inl 
=&\; ([f_1,g_1]\lor [f_2,g_2])\komp\eta\inl \\
=&\; [f_1,g_1]\komp\eta\inl\lor [f_2,g_2]\komp\eta\inl \\
=&\; f_1\lor f_2\\
=&\; [f_1\lor f_2,g_1\lor g_2]\comp\inl
\end{flalign*}
The second equation is shown analogously.
\end{proof}

\begin{proof}[of~\autoref{thm:elg-klee}]
%
%
We modify the claim slightly by replacing Clause~\itref{it:klee-eq2}. with the stronger
\begin{enumerate}[wide]
  \item[2$'$.] The Kleisli category of $\BBT$ is enriched over bounded 
  join-semilattices and strict join-preserving morphisms, and $\div=(\eta\inr)^\istar\c X\to TY$
  is the least element of $\BC(X,TY)$.
\end{enumerate}
Let us show that~\itref{it:klee-eq1}.--\itref{it:klee-eq3}. entail 2$'$. 
\begin{itemize}[wide]
  \item \emph{Right strictness of Kleisli composition}: $f\komp\div =\div$. Using naturality,
  $f\komp\div = f\komp (\eta\inr)^{\istar} = ([\eta\inl\komp f, \eta\inr] \komp\eta\inr)^{\istar}= (\eta\inr)^{\istar} =\div$.
  \item \emph{Left strictness of Kleisli composition}: $\div\komp f =\div$. Since
  $\eta\inr \komp f = {[\eta\inl,\eta\inr\komp f]} \komp \eta\inr$, by strong uniformity,
  $\div\komp f = (\eta\inr)^{\istar}\komp f = (\eta\inr)^{\istar} =\div$.
  \item \emph{$\div$ is the least element}, equivalently, $f\lor\div = f$ for all suitably 
  typed $f$. It suffices to consider the special case $f=\eta$, for then 
  $f\lor\div = f\komp (\eta\lor\div) = f\komp\eta = f$ for a general $f$. 
  
  Note that $(\eta\inl\lor \eta\inr) \komp (\eta\lor\div) = {[\eta\inl,\eta\inr\komp (\eta\lor\div)]} \komp (\eta\inl\lor \eta\inr)$,
  which by~\itref{it:klee-eq4}. and~\itref{it:klee-eq3}. entails $\eta\lor\div = (\eta\inl\lor \eta\inr)^\istar \komp (\eta\lor\div) = (\eta\inl\lor \eta\inr)^\istar = \eta$.
\end{itemize}

Now, given~$(\argument)^\istar$ of an Elgot monad, whose Kleisli category is enriched 
over join-semilattices, let 
\begin{displaymath}
  (f\c X\to TX)^\rstar = \bigl(\eta\inl\lor \eta\inr\komp f\c X\to T(X+X)\bigr)^\istar. 
\end{displaymath}
Conversely, given $(\argument)^\rstar$ 
of a Kleene monad, let 
\begin{displaymath}
  (f\c X\to T(Y+X))^\istar = ([\eta,\div]\komp f)\komp\bigl([\div,\eta]\komp f\c X\to TX\bigr)^\rstar.
\end{displaymath} 
We are left to 
check that these transformations are mutually inverse and that the expected properties
of defined operators are satisfied.

\medskip\noindent
(i)~ $(\argument)^\istar\to (\argument)^\rstar\to (\argument)^\istar$: Given 
  $f\c X\to T(Y+X)$, we need to show that
  \begin{align*}
  ([\eta,\div]\komp f)\komp (\eta\inl\lor \eta\inr\komp [\div,\eta]\komp f)^\istar = f^\istar.
  \end{align*}
  Indeed,
  \begin{flalign*}
  &&([\eta,\div]\komp& f)\komp (\eta\inl\lor \eta\inr\komp [\div,\eta]\komp f)^\istar\\
  &&=&\; \bigl([\eta\inl\komp [\eta,\div]\komp f, \eta\inr]^{\klstar} \comp (\eta\inl\lor \eta\inr\komp [\div,\eta]\komp f)\bigr)^\istar&\by{\NAT}\\
  &&=&\; ([\eta\inl,\div]\komp f\lor [\div,\eta\inr]\komp f)^\istar\\
  &&=&\; ([\eta\inl\lor\div,\div\lor\eta\inr]\komp f)^\istar&\by{\autoref{lem:klee_lem}}\\
  &&=&\; ([\eta\inl,\eta\inr]\komp f)^\istar\\
  &&=&\; f^\istar.
  \end{flalign*}
\medskip\noindent
(ii)~ $(\argument)^\rstar\to (\argument)^\istar\to (\argument)^\rstar$: Given 
  $f\c X\to TX$, we need to show that
  \begin{align*}
  ([\eta,\div]\komp (\eta\inl\lor \eta\inr\komp f))\komp([\div,\eta]\komp (\eta\inl\lor \eta\inr\komp f))^\rstar = f^\rstar.
  \end{align*}
  Indeed, $[\eta,\div]\komp (\eta\inl\lor \eta\inr\komp f) = \eta\lor\div = \eta$,
  and $[\div,\eta]\komp (\eta\inl\lor \eta\inr\komp f) = {\div\lor\eta\komp f= f}$,
  and therefore the right-hand side reduces to $\eta\komp f^\rstar = f^\rstar$.

Next, we show that from an Elgot monad we obtain a Kleene monad and back.
\medskip\noindent
(iii)~ \emph{From Elgot to Kleene}: We verify conditions from~\autoref{pro:klee-eq}.
  Since enrichment in semilattices is assumed, it suffices to check properties~\itref{it:klee-eq1}.--\itref{it:klee-eq3}.
  \begin{enumerate}[wide]
    \item We have
  %
  \begin{flalign*}
  && \eta\lor f^\rstar\komp f
      &\;= [\eta,f^\rstar]\komp (\eta\inl\lor \eta\inr\komp f)&\\
  &&  &\;= [\eta,(\eta\inl\lor \eta\inr\komp f)^\istar]\komp (\eta\inl\lor \eta\inr\komp f)&\\
  &&  &\;= (\eta\inl\lor\eta\inr\komp f)^\istar&\by{\FIX}\\
  &&  &\;= f^\rstar.
  \end{flalign*}
    \item $\eta^\rstar = (\eta\inl\lor\eta\inr)^\istar = \eta$ by the 
    global assumption of the theorem.
    \item We prove a stronger property $(f\lor\eta)^\rstar = f^\rstar$:
    \begin{flalign*}
    && (f\lor\eta)^\rstar &\;= (\eta\inl\lor\eta\inr\komp f\lor \eta\inr)^\istar &\\
    &&  &\;= (T[\id,\inr]\comp(\eta\inl\inl\lor\eta\inl\inr\komp f\lor\eta\inr))^\istar \\
    &&  &\;= (\eta\inl\inl\lor\eta\inl\inr\komp f\lor\eta\inr)^{\istar\istar} &\by{\COD}\\
    &&  &\;= ([\eta\inl\komp(\eta\inl\lor\eta\inr\komp f),\eta\inr]\komp(\eta\inl\lor\eta\inr))^{\istar\istar} \\
    &&  &\;= ((\eta\inl\lor\eta\inr\komp f)\komp (\eta\inl\lor\eta\inr)^{\istar})^{\istar} &\by{\NAT}\\
    &&  &\;= (\eta\inl\lor\eta\inr\komp f)^{\istar} &\by{assumption} \\
    &&  &\;= f^\rstar.
    \end{flalign*}
    \item Suppose that $f\komp h = g\komp f$. Then $(\eta\inl\lor\eta\inr\komp g)\komp f = \eta\inl\komp f\lor\eta\inr f\komp h= [\eta\inl,\eta\inr\komp f]\komp(\eta\inl\komp f\lor\eta\inr \komp h)$, which entails $(\eta\inl\lor\eta\inr\komp g)^\istar\komp f = (\eta\inl\komp f\lor\eta\inr \komp h)^\istar$,
    by strong uniformity.
    \begin{flalign*}
    && f\komp h^\rstar &\;= f\komp (\eta\inl\lor\eta\inr\komp h)^{\istar} &\\
    &&  &\;= (\eta\inl\komp f\lor\eta\inr\komp h)^{\istar}&\by{\NAT}\\
    &&  &\;= (\eta\inl\lor\eta\inr\komp g)^{\istar}\komp f&\by{\SUN}\\
    &&  &\;= g^\rstar\komp f
    \end{flalign*}
    
  \end{enumerate} 
\medskip\noindent
(iv)~ \emph{From Kleene to Elgot}: We need to verify the axioms of Elgot monads,
with \UNI replaced by \SUN.
%

  \medskip\noindent\FIX~: Given $f\c X\to T(Y+X)$,
  \begin{flalign*}
  && [\eta, f^{\istar}]\komp \comp f 
  &\;=  ([\eta\lor\div,\div\lor f^\istar]\komp f) \\*
  &&  &\;= [\eta,\div]\komp f\lor [\div,f^\istar]\komp f &\by{\autoref{lem:klee_lem}}\\
  &&  &\;= [\eta,\div]\komp f\lor f^\istar\komp [\div,\eta]\komp f \\
  &&  &\;= ([\eta,\div]\komp f)\komp (\eta\lor ([\div,\eta]\komp f)^\rstar\komp [\div,\eta]\komp f) \\
  &&  &\;= ([\eta,\div]\komp f)\komp ([\div,\eta]\komp f)^\rstar \\
  &&  &\;= f^{\istar}.
  \end{flalign*}
  \noindent\NAT:~ Given $g\c Y \to TZ$, $f\c X\to T(Y+X)$,
  \begin{flalign*}
  && g\komp f^{\istar} &\;= g\komp ([\eta,\div]\komp f)\komp ([\div,\eta]\komp f)^\rstar  &\\*
  &&  &\;= [g, \div] \komp f\komp ([\div,\eta]\komp f)^\rstar\\  
  &&  &\;= [\eta,\div]\komp [\eta\inl\komp g, \eta\inr] \komp f\komp ([\div,\eta]\komp [\eta\inl\komp g, \eta\inr] \komp f)^\rstar\\  
  &&  &\;= ([\eta\inl\komp g, \eta\inr] \komp f)^{\istar}.
  \end{flalign*}
  \noindent\COD:~ Let $f\c X \to T((Y + X) + X)$ and show $(T[\id,\inr] \comp f)^{\istar} = f^{\istar\istar}$. 
  Let $g = [[\div,\eta],\div]\komp f\c X\to TX$, $h=[\div,\eta]\komp \comp f\c X\to TX$. Then
  \begin{flalign*}
  && (T[\id,\inr] \comp f)^{\istar} &\;= ([\eta,\div]\komp T[\id,\inr] \comp f)\komp ([\div,\eta]\komp T[\id,\inr] \comp f)^\rstar &\\*
  &&  &\;=  [[\eta,\div],\div]\komp f\komp ([[\div,\eta],\eta]\komp \comp f)^\rstar \\
  &&  &\;=  [[\eta,\div],\div]\komp f\komp ([\div,\eta]\komp \comp f\lor [[\div,\eta],\div]\komp \comp f)^\rstar \\
  &&  &\;=  [[\eta,\div],\div]\komp f\komp (h\lor g)^\rstar \\
  &&  &\;=  [[\eta,\div],\div]\komp f\komp h^\rstar\komp (g\komp h^\rstar)^\rstar&\by{\autoref{lem:star-sum}} \\
  &&  &\;=  [\eta,\div]\komp [\eta,\div]\komp f\komp h^{\rstar}\komp ([\div,\eta]\komp [\eta,\div]\komp f\komp h^{\rstar})^\rstar\\
  &&  &\;=  ([\eta,\div]\komp f^{\istar})\komp ([\div,\eta]\komp f^{\istar})^\rstar\\
  &&  &\;=  f^{\istar\istar}.
  \end{flalign*}
  \noindent\SUN:~ Assume $f \komp h = [\eta\inl,\eta\inr\komp h] \komp g$. This entails 
  $([\div,\eta]\komp f) \komp h = ([\div,\eta]\komp [\eta\inl,\eta\inr\komp h]) \comp g = [\div, h]\komp g = h\komp [\div,\eta]\komp g$.
  By Lemma~\ref{pro:klee-eq}, 
  \begin{align*}
  ([\div,\eta]\komp f)^\rstar \komp h = h\komp ([\div,\eta]\komp g)^\rstar.
  \end{align*}
  We now have
  \begin{flalign*}
  && f^{\istar}\komp h &\;= ([\eta,\div]\komp f)\komp ([\div,\eta]\komp f)^\rstar\komp h  &\\
  &&  &\;= [\eta,\div]\komp f\komp h\komp ([\div,\eta]\komp g)^\rstar\\  
  &&  &\;= [\eta,\div]\komp [\eta\inl,\eta\inr\komp h] \komp g \komp ([\div,\eta]\komp g)^\rstar\\  
  &&  &\;= ([\eta,\div]\komp g)\komp ([\div,\eta]\komp g)^\rstar\\
  &&  &\;= g^{\istar}. 
  \end{flalign*}
This concludes the proof.
\end{proof}

In presence of assumptions~\itref{it:klee-eq1}.--\itref{it:klee-eq4}., the distinction 
between \UNI and \SUN becomes very subtle. 
%
\begin{example}[Filter Monad]
There is an Elgot monad $\BBT$, whose Kleisli category is enriched over bounded semilattices,
$(\eta\inl\lor\eta\inr)^\istar = \eta$, but $\BBT$ fails strong uniformity.
We prove it by adapting Kozen's separating example for left-handed and right-handed Kleene 
algebras~\cite[Proposition 7]{Kozen90}.

Recall that the \emph{filter monad}~\cite{Day75} sends every $X$ to the set of all filters on $X$,
equivalently to those maps $h\c (X\to 2)\to 2$, which preserve $\top$ and $\land$:
$h(\top) = \top$, $h(f\land g) = h(f)\land h(g)$ where $\top$ and $\land$ on $X\to 2$
are computed pointwise. For us, it will be more convenient to use the equivalent formulation,
obtained by flipping the order on $2$ (so, the resulting monad $\BBT$ could be actually 
called the \emph{ideal monad}). Every $TX$ is then the set of those  
$h\c\PSet X\to 2$, for which
\begin{align*}
f(\emptyset) = \bot, && f(s\cup t) = f(s)\lor f(t).
\end{align*}
\begin{enumerate}[wide]
  \item Note that Kleisli category $\Set_{\BBT}$ is dually isomorphic to a category $\BC$,
  for which every $\BC(X,Y)$ consists of functions $\PSet X\to\PSet Y$, preserving finite joins
  (in particular, monotone).
  This category has finite products: $\PSet\iobj$ is the terminal object and $\PSet X\times\PSet Y = \PSet(X+Y)$,
  by definition.
  \item Under this dual isomorphism, every morphism $f\c X\to T(Y+X)$ corresponds to a morphism $\hat f\c\PSet Y\times\PSet X\to\PSet X$
  in $\BC$ where we compute a fixpoint $\PSet Y\to\PSet X$ using the \emph{Knaster-Tarski theorem}, and transfer
  it back to $\BC$ as $f^\istar\c X\to TY$. 
  \item The construction of $f^\istar$ entails
  both $(\eta\inl\lor\eta\inr)^\istar \leq \eta$ and $(\eta\inl\lor\eta\inr)^\istar \geq \eta$,
  hence $(\eta\inl\lor\eta\inr)^\istar = \eta$.
  \item Enrichment in semilattices is obvious in view of the dual isomorphism of $\Set_{\BBT}$
  and~$\BC$.
  \item The \FIX law follows by construction. The remaining Elgot monad laws follow by \emph{transfinite induction}. 
  \item If $\BBT$ was a Kleene monad, any $\BC(X,X)$ would be a Kleene algebra,
  but Kozen showed that it is not, hence $\BBT$ is not a Kleene monad.
  \item By \autoref{thm:elg-klee}, $\BBT$ fails \SUN.
\end{enumerate}
\end{example}


\section{Conclusions}\label{sec:conc}
When it comes to modelling and semantics, many issues can be framed and treated
in terms of universal algebra and coalgebra. However, certain phenomena, such 
as recursion, partiality, extensionality, require additional structures, often 
imported from the theory of complete partial orders, by enriching categories and 
functors, and devising suitable structures, such as recursion and more specifically 
iteration. In many settings though, iteration is sufficient, and can be treated as a self-contained
ingredient whose properties matter, while a particular construction behind it 
does not. 
From this perspective, Elgot monads present a base fundamental 
building block in semantics. 

We formally compared Elgot monads with Kleene monads, which are 
a modest generalization of Kleene algebras. In contrast to inherently categorical 
Elgot monads, Kleene algebra is a simple notion, couched in traditional algebraic terms. 
The price of this simplicity is a tight pack of laws, which must be accepted altogether, but 
which are well-known to be conflicting with many models of iteration. We proposed a novel notion of \emph{while-monad},
which in the categorical context are essentially equivalent to Elgot monads,
and yet while-monads are morally a three-sorted algebra over (Boolean) 
decisions, programs and certain well-behaved programs (figuring in the 
so-called uniformity principle). This is somewhat similar to the extension 
of Kleene algebra with \emph{tests}~\cite{KozenSmith96}. The resulting \emph{Kleene algebra 
with tests} is two-sorted, with tests being a subsort of programs, 
and forming a Boolean algebra. Our \emph{decisions} unlike tests do not form 
a subsort of programs, but they do support operations of Boolean algebra, without 
however complying with all the Boolean algebra laws.  
We have then related Elgot monads (and while-monads)
with Kleene monads, and as a side-effect produced a novel axiomatization of Kleene 
algebra (\autoref{pro:klee-eq}), based on a version of the uniformity principle. We regard the present 
work as a step towards bringing the gap between Elgot iteration and Kleene iteration,
not only in technical sense, but also in the sense of concrete usage scenarios.
We plan to further explore algebraic axiomatizations of iteration, based on the 
current axiomatization of while-monads.

\bibliographystyle{plain}
\bibliography{monads}

\end{document}